\documentclass[10pt]{article}
\usepackage{graphicx}
\usepackage[T1]{fontenc}
\usepackage[utf8]{inputenc}
\usepackage{amsmath,amsthm} 
\usepackage{amssymb}
\usepackage{amsfonts}
\usepackage{dsfont}
\usepackage{mathtools}
\usepackage{amsthm}
\usepackage{amsmath}
\usepackage{textcomp}
\usepackage{xcolor}
\usepackage{color,soul}
\usepackage{eurosym}
\usepackage{chngcntr}
\usepackage{numprint}
\usepackage{ulem}

\counterwithout{table}{section}
\counterwithout{table}{subsection}

\usepackage{float}
\usepackage{graphicx}
\usepackage{caption}
\usepackage{subcaption}

\usepackage{geometry}
\newtheorem{thm}{Theorem}
\newtheorem{crl}{Corollary}
\newtheorem{prop}{Proposition}
\newtheorem{defi}{Definition}
\newtheorem{rem}{Remark}

\begin{document}

\title{Multi-asset optimal execution and statistical arbitrage strategies under Ornstein-Uhlenbeck dynamics}

\author{Philippe \textsc{Bergault}\footnote{Université Paris 1 Panthéon-Sorbonne, Centre d'Economie de la Sorbonne, 106 Boulevard de l'Hôpital, 75642 Paris Cedex 13, France, philippe.bergault@etu.univ-paris1.fr.} \and Fayçal \textsc{Drissi}\footnote{Université Paris 1 Panthéon-Sorbonne, Centre d'Economie de la Sorbonne, 106 Boulevard de l'Hôpital, 75642 Paris Cedex 13, France, faycal.drissi@etu.univ-paris1.fr.} \and Olivier \textsc{Guéant}\footnote{Université Paris 1 Panthéon-Sorbonne, Centre d'Economie de la Sorbonne, 106 Boulevard de l'Hôpital, 75642 Paris Cedex 13, France, olivier.gueant@univ-paris1.fr. Member of the Louis Bachelier Fellows. \textit{Corresponding author.}}}
\date{}
\maketitle

\begin{abstract}
In recent years, academics, regulators, and market practitioners have increasingly addressed liquidity issues. Amongst the numerous problems addressed, the optimal execution of large orders is probably the one that has attracted the most research works, mainly in the case of single-asset portfolios. In practice, however, optimal execution problems often involve large portfolios comprising numerous assets, and models should consequently account for risks at the portfolio level. In this paper, we address multi-asset optimal execution in a model where prices have multivariate Ornstein-Uhlenbeck dynamics and where the agent maximizes the expected (exponential) utility of her PnL. We use the tools of stochastic optimal control and simplify the initial multidimensional Hamilton-Jacobi-Bellman equation into a system of ordinary differential equations (ODEs) involving a Matrix Riccati ODE for which classical existence theorems do not apply. By using \textit{a priori} estimates obtained thanks to optimal control tools, we nevertheless prove an existence and uniqueness result for the latter ODE, and then deduce a verification theorem that provides a rigorous solution to the execution problem. Using examples based on data from the foreign exchange and stock markets, we eventually illustrate our results and discuss their implications for both optimal execution and statistical arbitrage.    
\end{abstract}

\vspace{8mm}

\setlength\parindent{0pt}

\textbf{Key words:} Optimal execution, Statistical arbitrage, Stochastic optimal control, Riccati equations.\\

\vspace{5mm}

\section{Introduction}

When executing large blocks of assets, financial agents need to control their overall trading costs by finding the optimal balance between trading rapidly to mitigate market price risk and trading slowly to minimize execution costs and market impact. Building on the first rigorous approaches introduced by Bertsimas and Lo in \cite{bertsimas1998optimal} and Almgren and Chriss in \cite{almgren1999value} and \cite{almgren2001optimal}, many models for the optimal execution of large orders have been proposed in the last two decades. Subsequently, almost all practitioners today slice their large orders into small (child) orders according to optimized trading schedules inspired by the academic literature.\\

The basic Almgren-Chriss model is a discrete-time model where the agent posts market orders (MOs) to maximize a mean-variance objective function. Many extensions of this seminal model have been proposed. Regarding the framework, (Forsyth and Kennedy, \cite{forsyth2012optimal}) examines the use of quadratic variation rather than variance in the objective function, (Schied and Sch{\"o}neborn, \cite{schied2009risk}) uses stochastic optimal control tools to characterize and find optimal strategies for a Von Neumann–Morgenstern investor, and (Guéant, \cite{gueant2015optimal}) provides results for optimal liquidation within a Von Neumann-Morgenstern expected utility framework with general market impact functions and derives subsequent results for block trade pricing. As for the model parameters, (Almgren, \cite{almgren2003optimal}) studies the case of random execution costs, (Almgren, \cite{almgren2009optimal, almgren2012optimal}) addresses stochastic liquidity and volatility, (Lehalle, \cite{lehalle2008rigorous}) discusses how to take into account statistical aspects of the variability of estimators of the main exogenous variables such as volumes or volatilities in the optimization phase, and (Cartea and Jaimungal, \cite{cartea2016incorporating}) provides a closed-form strategy incorporating order flows from all agents. Furthermore, numerous market impact and limit order book (LOB) models have also been studied. For instance, (Lorenz and Schied,~\cite{lorenz2013drift}) investigates the stability of optimal strategies in the presence of transient price impact with exponential decay and general dynamics of a drift in the underlying price process accounting for the price impact of other agents. (Obizhaeva and Wang, \cite{obizhaeva2013optimal}), later generalized in (Alfonsi, Fruth, and Schied, \cite{alfonsi2008constrained}), proposes a single-asset market impact model where price dynamics are derived from a dynamic LOB model with resilience, (Alfonsi and Schied, \cite{alfonsi2010optimal}) derives explicit optimal execution strategies in a discrete-time LOB model with general shape functions and an exponentially decaying price impact, (Gatheral, \cite{gatheral2010no}) uses the no-dynamic-arbitrage principle to address the viability of market impact models, and (Gatheral, Schied, and Slynko, \cite{gatheral2012transient}) obtains explicit optimal strategies with a transient market impact in an expected cost minimization setup. As for order and execution strategy types, the Almgren-Chriss framework focuses on orders of the Implementation Shortfall (IS) type with MOs only. Other execution strategies have been studied in the literature, like Volume-Weighted Average Price (VWAP) orders in (Konishi, \cite{konishi2002optimal}), (Frei and Westray, \cite{frei2015optimal}) and (Gu{\'e}ant and Royer, \cite{gueant2014vwap}), but also Target Close (TC) orders and Percentage of Volume (POV) orders, in (Gu{\'e}ant, \cite{gueant2016financial}). Besides, several models focusing on optimal execution with limit orders have been proposed, as in (Bayraktar and Ludkovski,~\cite{bayraktar2014liquidation}), but also in (Gu{\'e}ant, Lehalle, and Fernandez-Tapia, \cite{gueant2012optimal}) and (Gu{\'e}ant and Lehalle, \cite{gueant2015general}). Regarding the existence of several venues, the case of optimal splitting of orders across different liquidity pools has been addressed in (Laruelle, Lehalle, and Pages, \cite{laruelle2011optimal}), in (Cartea, Jaimungal, and Penalva, \cite{cartea2015algorithmic}), and more recently in (Baldacci and Manziuk, \cite{baldacci2020adaptive}).\\

Another recent and important stream of the optimal execution literature deals with adding predictive signals of future price changes.\footnote{We consider this stream of the literature to be closely related to our topic of multi-asset optimal execution. Indeed, when trading an asset, the dynamics of another asset within or outside the portfolio can be regarded as a predictive signal that can enhance the execution process.} Typical examples of these signals include a drift in asset prices, order book imbalances, forecasts of the future order flow of market participants, and other price-based technical indicators. The usual formalism in the literature with predictive signals is to consider Brownian or Black-Scholes dynamics, along with independent mean-reverting Markov signals. The case of Ornstein-Uhlenbeck-type signals is of special interest as it usually leads to closed-form formulas. For the interested reader, we refer to (Belak and Muhle-Karbe, \cite{belak2018optimal}) where the authors consider optimal execution with general Markov signals and an application to ``target zone models'', and to (Lehalle and Neuman, \cite{lehalle2019incorporating}) and (Neuman and Vo{\ss}, \cite{neuman2020optimal}) in which the authors provide an optimal trading framework incorporating Markov signals and a transient market impact.\\

In practice, operators routinely face the problem of having to execute simultaneously large orders regarding various assets, such as in block trading for funds facing large subscriptions or withdrawals, or when considering multi-asset trades in statistical arbitrage trading strategies. More generally, banks and market makers manage their (il)liquidity and market risk, when it comes to executing trades, in the context of a central risk book; hence the need for multi-asset models. However, in contrast to the single-asset case, the existing literature on the joint execution scheduling of large orders in multiple assets, or a single asset inside a multi-asset portfolio, is rather limited. Besides, most existing papers simply consider correlated Brownian motions when modelling the joint dynamics of prices. The problem of using single-asset models or unrealistic multivariate models for portfolio trading is that the resulting trading curves of individual assets do not balance well execution costs / market impact with price risk at the portfolio or strategy level.\\

The first paper presenting a way to build multi-asset trading curves in an optimized way is (Almgren and Chriss,~\cite{almgren2001optimal}). Almgren and Chriss considered indeed, in an appendix of their seminal paper, a multi-asset extension of their discrete-time model. A few extensions to this model have been proposed since then. (Lehalle, \cite{lehalle2009rigorous}) considers adding an inventory constraint to balance the different portfolio lines during the portfolio execution process. (Schied and Sch{\"o}neborn, \cite{schied2010optimal}) shows that when prices follow Bachelier dynamics, deterministic strategies are optimal for a trader with an exponential utility objective function. In (Cartea, Jaimungal, and Penalva, \cite{cartea2015algorithmic}), the authors use stochastic control tools to derive optimal execution strategies for basic multi-asset trading algorithms such as optimal entry/exit times and cointegration-based statistical arbitrage. (Bismuth, Gu{\'e}ant, and Pu, \cite{bismuth2019portfolio}) addresses optimal portfolio liquidation (along with other problems) by coupling Bayesian learning and stochastic control to derive optimal strategies under uncertainty on model parameters in the Almgren-Chriss framework. Regarding the literature around the addition of predictive signals, (Emschwiller, Petit, and Bouchaud, \cite{emschwiller2021optimal}) extends optimal trading with Markovian predictors to the multi-asset case, with linear trading costs, using a mean-field approach that reduces the problem to a single-asset one. \\

A classical model for the multivariate dynamics of financial variables that goes beyond that of correlated Brownian motions is the multivariate Ornstein-Uhlenbeck (multi-OU) model. It is especially attractive because it is parsimonious, and yet general enough to cover a wide spectrum of multi-dimensional dynamics. Multi-OU dynamics offer indeed a large coverage since particular cases include correlated Brownian motions but also cointegrated dynamics which are heavily used in statistical arbitrage. (Cartea, Gan, and Jaimungal, \cite{cartea2019trading}) is, to our knowledge, the pioneering paper in the use of the multi-OU model for the price dynamics in a multi-asset optimal execution problem. Indeed, the authors proposed an interesting model where the asset prices have multi-OU dynamics and the agent maximizes an objective function given by the expectation of the Profit and Loss (PnL) minus a running penalty related to the instantaneous variance of the portfolio. In their approach, the problem boils down to a system of ODEs involving a Matrix Riccati ODE for which the classical existence theorems related to linear-quadratic control theory apply.\\

In this paper, we propose a model similar to the one in \cite{cartea2019trading}, but where the objective function is of the Von Neumann-Morgenstern type: an expected exponential utility of the PnL.\footnote{Our model accounts therefore for the risk in a different manner than the model presented in \cite{cartea2015algorithmic}. Comparisons are difficult to carry out as risk aversion parameters in the two models have different meanings.} By using classical stochastic optimal control tools we show that the problem boils down to solving a system of ODEs involving a Matrix Riccati ODE. However, unlike what happens in \cite{cartea2019trading}, the use of an expected exponential utility framework to account for the risk leads to a Matrix Riccati ODE for which classical existence theorems do not apply. By using \textit{a priori} estimates obtained thanks to optimal control tools, we nevertheless prove an existence and uniqueness result for the latter ODE and obtain a verification theorem that provides a rigorous solution to the execution problem.\\

The main contribution of this paper is therefore to propose a model for multi-asset portfolio execution under multi-OU price dynamics in an expected utility framework that accounts for the overall risk associated with the execution process. We focus on the problem where an agent is in charge of unwinding a large portfolio, but also illustrate the use of our results for multi-asset statistical arbitrage purposes.\\

The remainder of this paper is organized as follows. In Section 2 we present the optimal execution problem in the form of a stochastic optimal control problem and show that solving the associated Hamilton-Jacobi-Bellman (HJB) equation boils down to solving a system of ODEs involving a Matrix Riccati ODE. We then prove a global existence result for that ODE and eventually provide a solution to the initial stochastic optimal control problem thanks to a verification argument. In Section 3, we then illustrate our results on both real data, from the foreign exchange and stock markets, and simulated data. Our examples focus on optimal liquidation but we also illustrate and discuss the use of our results for building statistical arbitrage strategies. The core of the paper is followed by two appendices: one dedicated to the special case where the multi-OU dynamics reduce to simple correlated Brownian dynamics and another dedicated to some form of limit case where execution costs and terminal penalty are ignored -- that limit case being useful to obtain \textit{a priori} estimates for our general problem.\\

\section{The optimal liquidation problem}

\subsection{Modelling framework and notations}
In this paper, we consider a filtered probability space $\left(\Omega, \mathcal F, \mathbb P; \mathbb F = (\mathcal F_t)_{t \in [0,T]} \right)$ satisfying the usual conditions. We assume this probability space to be large enough to support all the processes we introduce.\\

We consider a market with $d \in \mathbb N^*$ assets,\footnote{We denote by $\mathbb N^*$ the set $\mathbb N^* := \mathbb N \backslash \{0\}$ of positive integers.} and a trader wishing to liquidate her portfolio over a period of time $[0,T]$, with $T>0$. Her inventory process\footnote{The superscript $^\intercal$ designates the transpose operator. It transforms here a line vector into a column vector.} $(q_t)_{t \in [0,T]} = \left(q^1_t, \ldots, q^d_t \right)^\intercal _{t \in [0,T]}$ evolves as
\begin{align}
\label{qProcess}
    dq_t = v_t dt,
\end{align}
with $q_0 \in \mathbb R^d$ given, where $(v_t)_{t \in [0,T]}= (v^1_t, \ldots, v^d_t)^\intercal _{t \in [0,T]}$ represents the trading rate of the trader for each asset.\\

The fundamental prices of the $d$ assets are modelled as a $d$-dimensional Ornstein-Uhlenbeck process $(S_t)_{t\in [0,T]} = \left(S^1_t, \ldots, S^d_t \right)^\intercal_{t\in [0,T]}$:
\begin{align}
\label{PriceProces}
    dS_t &  = R(\overline S - S_t) dt + V dW_t, 
\end{align}
and we introduce the market price process 
$(\tilde S_t)_{t\in [0,T]} = \left(\tilde S^1_t, \ldots, \tilde S^d_t \right)^\intercal_{t\in [0,T]}$ with dynamics:
\begin{align}
\label{MPriceProces}
    d\tilde S_t &  = dS_t + Kv_t dt, 
\end{align}
with  $S_0= \tilde S_0 \in {\mathbb R}^d $ given, where $\overline S \in {\mathbb R}^d $, $R  \in \mathcal M_d (\mathbb R)$,  $V \in \mathcal M_{d,k}(\mathbb R)$, $K \in \mathcal S_d(\mathbb R)$,\footnote{We denote by $\mathcal M_{d,k}(\mathbb R)$ the set of $d\times k$ real matrices and by $\mathcal M_d (\mathbb R) := \mathcal M_{d,d}(\mathbb R)$ the set of $d \times d$ real square matrices. The set of real symmetric $d \times d$ matrices is denoted by $\mathcal S_d (\mathbb R)$.} and $(W_t)_{t \in [0,T]} = \left(W^1_t, \ldots, W^k_t \right)^\intercal_{t \in [0,T]}$ is a $k$-dimensional standard Brownian motion (with independent coordinates) for some $k \in \mathbb{N}^*$. In these dynamics, the matrix $R$ steers the deterministic part of the process, $\overline S$ represents the unconditional long-term expectation of $(S_t)_{t\in [0,T]}$, and $V$ drives the dispersion (for what follows, we introduce $\Sigma = VV^\intercal$ the covariation matrix of the process). The matrix $K$ represents the linear permanent impact the agent has on the prices.\footnote{It is assumed symmetric to avoid price manipulation.}  More precisely, since $$d\tilde S_t  = dS_t + Kv_t dt = R(\overline S - S_t) dt + Kv_t dt + V dW_t = R(\overline S + K(q_t-q_0) - \tilde S_t) dt + Kv_t dt + V dW_t,$$ trading impacts both current market prices and long-term expectations.\\  

Ornstein-Uhlenbeck processes are well suited when prices exhibit mean reversion and/or when there exist one or several linear combinations of asset prices that are stationary. In the latter case, we say that the assets involved in the linear combinations are cointegrated (a situation often encountered in statistical arbitrage). For more details on cointegration in continuous time, we refer to (Comte \cite{comte1999discrete}).\\

Finally, the process $(\tilde X_t)_{t \in [0,T] }$ modelling the trader's cash account has the dynamics
\begin{align}
\label{cashProcess1}
    d\tilde X_t = - v_t^\intercal \tilde  S_t dt - L(v_t)dt,
\end{align}
with $\tilde X_0 \in \mathbb R$ given, where $L: \mathbb R^d \rightarrow \mathbb R_+$ is a function representing the temporary market impact of trades and/or the execution costs incurred by the trader (see (Gu{\'e}ant, \cite{gueant2016financial}) for an introduction to this type of models). In this paper, we only consider the case where $L$ is a positive-definite quadratic form, i.e.\footnote{The subset of positive-definite and positive semi-definite matrices of $\mathcal S_d (\mathbb R)$ are respectively denoted by $\mathcal S^{++}_d (\mathbb R)$ and $\mathcal S^{+}_d (\mathbb R)$.} $$L(v) = v^\intercal \eta v \quad \text{with } \quad \eta \in \mathcal S^{++}_d (\mathbb R).$$

The trader aims at maximizing the expected utility of her wealth at the end of the trading window $[0, T]$. This wealth is the sum of the amount $\tilde X_T$ on the cash account at time $T$ and the value of the remaining inventory evaluated here at $q_T^\intercal \tilde S_T - \tilde \ell(q_T)$, where the discount term $\tilde \ell(q_T)$ applied to the Mark-to-Market (MtM) value of the remaining assets ($q_T^\intercal \tilde S_T$) penalizes
any non-zero terminal position.\footnote{This penalization term relaxes the hard constraint that imposes $q_T = 0$ in some liquidation problems. This relaxation enables to use classical tools of optimal control while the problem with hard constraint (not addressed in this paper) features a singular boundary condition that makes the problem more difficult to address mathematically.} In this paper, we only consider the case where $\tilde \ell$ is a positive-definite quadratic form, i.e. $\tilde \ell(q) = q^\intercal \tilde \Gamma q$ with $\tilde \Gamma \in \mathcal S_d^{++} (\mathbb R)$ (see below for a stronger assumption on $\tilde \Gamma$).\\

To define the set of admissible controls $\mathcal A$, we first introduce a notion of ``linear growth'' relevant in our context.
\begin{defi}
\label{defiLinearGrowth}
Let $t \in [0,T].$ An $\mathbb R^d$-valued, $\mathbb F$-adapted process $(\zeta_s)_{s \in [t,T]}$ is said to satisfy a linear growth condition on $[t,T]$ with respect to $(S_s)_{s \in [t,T]}$ if there exists a constant $C_{t,T}>0$ such that for all $s \in [t,T]$,
$$\| \zeta_s \| \le C_{t,T} \left(1 + \underset{\tau \in [t,s]}{\sup} \| S_\tau \| \right)\vspace{-0.2cm}$$
almost surely.\footnote{In all this paper, $\|.\|$ denotes a fixed norm on $\mathbb R^d$ (for instance, the Euclidean norm).}
\end{defi}
We then define for all $t \in [0,T]$:
\begin{align}
\label{AdmissibleSet_t}
\mathcal A_t = \left\{ (v_s)_{s \in [t,T]},\ \mathbb R^d\textrm{-valued},\ \mathbb F\textrm{-adapted, satisfying a linear growth condition with respect to } (S_s)_{s \in [t,T]}  \right\},    
\end{align}
and take $\mathcal A := \mathcal A_0.$\footnote{We restrict our analysis to linear growth strategies for mathematical convenience, but we expect the candidate control to be optimal among a larger class of processes.}\\

Mathematically, the trader therefore wants to solve the dynamic optimization problem
\begin{align}
\label{problem1}
\underset{v \in \mathcal A}{\sup}\ \mathbb E \left[ - e^{-\gamma \left( \tilde X_T + q_T^\intercal \tilde S_T - \tilde \ell(q_T) \right)} \right],
\end{align}
where $\gamma>0$ is the absolute risk aversion parameter of the trader.\\

Notice that
\begin{align*}
\tilde X_T + q_T^\intercal \tilde S_T - \tilde \ell(q_T) &= \tilde X_0 + q_0^\intercal \tilde S_0 + \int_0^T q_t^\intercal d\tilde S_t  - \int_0^T L(v_t)dt- \tilde \ell(q_T)\\
&= X_0 + q_0^\intercal S_0 + \int_0^T q_t^\intercal d S_t + \int_0^T q_t^\intercal   Kv_t dt - \int_0^T L(v_t)dt - \tilde \ell(q_T) \\
&= X_T + q_T ^\intercal S_T - \tilde \ell(q_T) + \frac 12 q_T^\intercal Kq_T - \frac 12 q_0^\intercal Kq_0,
\end{align*}
where $X_0 = \tilde X_0$ and the process $(X_t)_{t \in [0,T] }$ has dynamics
\begin{align}
\label{cashProcess2}
    d X_t = - v_t^\intercal  S_t dt - L(v_t)dt.
\end{align}

Let us now define the penalty function $ \ell:\mathbb R^d \rightarrow \mathbb R$ by:
$$\ell(q) = \tilde \ell (q) - \frac 12 q^\intercal K q = q^\intercal \tilde \Gamma q - \frac 12 q^\intercal K q  \quad \forall q \in \mathbb R^d.$$
In what follows, we assume that $\ell$ is a positive semi-definite quadratic form, i.e. $\Gamma = \tilde \Gamma  - \frac 12 K \in \mathcal S_d^{+}(\mathbb R)$.

\begin{rem}
The assumption on $\tilde \Gamma$ is not restrictive as in practice, $\tilde \Gamma$ (and therefore $\Gamma$) is chosen arbitrarily large to enforce liquidation.\\
\end{rem}

It is then straightforward to see that Problem \eqref{problem1} is equivalent to the following problem:
\begin{align}
\label{problem}
\underset{v \in \mathcal A}{\sup}\ \mathbb E \left[ - e^{-\gamma \left(  X_T + q_T^\intercal  S_T -  \ell(q_T) \right)} \right].
\end{align}
It is natural to use the tools of stochastic optimal control to solve the above dynamic optimization problem. Let us define the value function of the problem $u:[0,T] \times \mathbb R \times \mathbb R^d \times \mathbb R^d \rightarrow \mathbb R$ as
\begin{align}
\label{vfdef}
u(t,x,q,S) = \underset{v \in \mathcal A_t}{\sup} \mathbb E \left[ - e^{-\gamma \left( X^{t,x,S,v}_T + (q^{t,q,v}_T)^\intercal S^{t,S}_T - \ell(q^{t,q,v}_T) \right)} \right],
\end{align}
where, for $(t,x,q,S) \in [0,T] \times \mathbb R \times \mathbb R^d \times \mathbb R^d$ and $v \in \mathcal A_t$, the processes $(q^{t,q,v}_s)_{s \in [t,T]}$, $(S^{t,S}_s)_{s \in [t,T]}$, and $(X^{t,x,S,v}_s)_{s \in [t,T]}$ have respective dynamics

$$\begin{array}{lcl}
dq^{t,q,v}_s &=& v_s ds, \\
 dS^{t,S}_s &=&  R(\overline{S} - S^{t,S}_s) ds + V dW_s,\\
  dX^{t,x,S,v}_s &=& - v_s^\intercal S^{t,S}_s ds - L(v_s)ds,
\end{array}
$$
with $S^{t,S}_t = S$, $q^{t,q,v}_t = q$, and $X^{t,x,S,v}_t = x$.

\subsection{Hamilton-Jacobi-Bellman equation}

The HJB equation associated with Problem \eqref{problem} is given by\footnote{$u$ will be solution to that equation, but as we do not know it yet, we write the equation with an unknown function $w$.}
\begin{align}
0 =\ &\partial_t w(t,x,q,S) + \underset{v \in \mathbb R^d}{\sup} \left( - (v^\intercal S+L(v))\partial_x w(t,x,q,S) + v^\intercal \nabla_q w(t,x,q,S) \right)\nonumber\\
& + (\overline S-S)^\intercal R^\intercal  \nabla_S w(t,x,q,S) + \frac 12 \textrm{Tr} \left( \Sigma D^2_{SS} w(t,x,q,S ) \right),
\label{HJBu}
\end{align}
for all $(t,x,q,S) \in [0,T) \times \mathbb R \times \mathbb R^d \times \mathbb R^d$ with the terminal condition 
\begin{align}
\label{termcondu}
w(T,x,q,S) = -e^{-\gamma \left(x+q^\intercal S - \ell(q) \right)} \quad \forall (x,q,S) \in \mathbb R \times \mathbb R^d \times \mathbb R^d.
\end{align}
In order to study \eqref{HJBu}, we are going to use the following ansatz:
\begin{align}
\label{ansatz1}
w(t,x,q,S) = -e^{-\gamma \left(x+q^\intercal S + \theta(t,q,S) \right)} \quad \forall (t,x,q,S) \in [0,T] \times \mathbb R \times \mathbb R^d \times \mathbb R^d.
\end{align}

The interest of this ansatz is based on the following proposition:

\begin{prop}
Let $\tau<T$. If there exists $\theta \in C^{1,1,2}([\tau ,T] \times \mathbb R^d \times \mathbb R^d, \mathbb R)$ solution to
\begin{align}
\label{HJBtheta0}
0 =\ & \partial_t \theta(t,q,S) + \underset{v \in \mathbb R^d}{\sup} \left(  v^\intercal \nabla_q \theta(t,q,S) -L(v) \right) +\frac 12 \textrm{Tr} \left(  \Sigma D^2_{SS}\theta(t,q,S) \right)\\
& - \frac{\gamma}{2}  (q + \nabla_S \theta(t,q,S))^\intercal \Sigma (q + \nabla_S \theta(t,q,S)) +  (\overline S-S)^\intercal R^\intercal  (\nabla_S \theta(t,q,S)+q) \nonumber
\end{align}
on $[\tau ,T) \times \mathbb R^d \times \mathbb R^d$, with terminal condition
\begin{align}
\label{termcondtheta0}
\theta(T,q,S) = - \ell(q) \quad \forall (q,S) \in \mathbb R^d \times \mathbb R^d,
\end{align}
then the function $w:[\tau,T] \times \mathbb R \times \mathbb R^d \times \mathbb R^d \rightarrow \mathbb R$ defined by 
$$w(t,x,q,S) = -e^{-\gamma \left(x+q^\intercal S + \theta(t,q,S) \right)} \quad \forall (t,x,q,S) \in [\tau,T] \times \mathbb R \times \mathbb R^d \times \mathbb R^d$$
is a solution to \eqref{HJBu} on $[\tau,T)\times \mathbb R \times \mathbb R^d \times \mathbb R^d$ with terminal condition \eqref{termcondu}.
\end{prop}
\begin{proof}
Let $\theta \in C^{1,1,2}([\tau ,T] \times \mathbb R^d \times \mathbb R^d, \mathbb R)$ be a solution to \eqref{HJBtheta0} on $[\tau, T)\times \mathbb R^d \times \mathbb R^d$ with terminal condition \eqref{termcondtheta0}, then we have for all $(t,x,q,S) \in [\tau,T) \times \mathbb R \times \mathbb R^d \times \mathbb R^d$:
\begin{align*}
& \partial_t w(t,x,q,S) + \underset{v \in \mathbb R^d}{\sup} \left( - (v^\intercal S+L(v))\partial_x w(t,x,q,S) + v^\intercal \nabla_q w(t,x,q,S) \right)\nonumber\\
& + (\overline S-S)^\intercal R^\intercal  \nabla_S w(t,x,q,S) + \frac 12 \textrm{Tr} \left( \Sigma D^2_{SS} w(t,x,q,S ) \right)  \\  
=\ & -\gamma \partial_t \theta(t,q,S)w(t,x,q,S) + \underset{v \in \mathbb R^d}{\sup} \left( \gamma (v^\intercal S+L(v)) w(t,x,q,S) -\gamma v^\intercal (\nabla_q \theta(t,q,S)+S) w(t,x,q,S) \right)\nonumber\\
& + \frac{\gamma^2}{2} \textrm{Tr}\left(\Sigma (q + \nabla_S \theta(t,q,S)) (q + \nabla_S \theta(t,q,S))^\intercal w(t,x,q,S)\right)\\
& - \gamma (\overline S-S)^\intercal R^\intercal  (\nabla_S \theta(t,q,S)+q) w(t,x,q,S) - \frac 12 \textrm{Tr} \left(\gamma  \Sigma D^2_{SS}\theta(t,q,S) w(t,x,q,S ) \right)  \\
=\ & -\gamma w(t,x,q,S) \Bigg(\partial_t \theta(t,q,S) + \underset{v \in \mathbb R^d}{\sup} \left(  v^\intercal \nabla_q \theta(t,q,S) -L(v) \right) +\frac 12 \textrm{Tr} \left(  \Sigma D^2_{SS}\theta(t,q,S) \right)\\
& - \frac{\gamma}{2}  (q + \nabla_S \theta(t,q,S))^\intercal \Sigma (q + \nabla_S \theta(t,q,S)) +  (\overline S-S)^\intercal R^\intercal  (\nabla_S \theta(t,q,S)+q)\Bigg)\\
=\ &0.
\end{align*}
As it is straightforward to verify that $w$ satisfies the terminal condition \eqref{termcondu}, the result is proved.
\end{proof}

The above result does not rely on the quadratic assumptions for $L$ and $\ell$. In the quadratic case we consider in this paper, $\theta$ can be found in almost closed form. To prove this point, the first thing to notice is that the Legendre-Fenchel transform of $L$ writes
\begin{align}
    H: p \in \mathbb R^d \mapsto \underset{v \in \mathbb R^d}{\sup} \ v^\intercal p - L(v) = \underset{v \in \mathbb R^d}{\sup} \ v^\intercal p - v^\intercal \eta v = \frac 14 p^\intercal \eta^{-1} p,
    \label{LegendreFenchel}
\end{align}
as the supremum is reached at  $v^* = \frac 12 \eta^{-1} p$.\\

Consequently, we get the following HJB equation for $\theta$:
\begin{align}
\label{HJBtheta2}
0 & = \partial_t \theta(t,q,S) + \frac 14 \nabla_q \theta(t,q,S)^\intercal \eta^{-1}\nabla_q \theta(t,q,S) +\frac 12 \textrm{Tr} \left(  \Sigma D^2_{SS}\theta(t,q,S) \right)\\
& - \frac{\gamma}{2}  (q + \nabla_S \theta(t,q,S))^\intercal \Sigma (q + \nabla_S \theta(t,q,S)) +  (\overline S-S)^\intercal R^\intercal  (\nabla_S \theta(t,q,S)+q), \nonumber
\end{align}
with terminal condition
\begin{align}
\label{termcondtheta1}
\theta(T,q,S) = - q^\intercal \Gamma q \quad \forall (q,S) \in \mathbb R^d \times \mathbb R^d.
\end{align}

To further study \eqref{HJBtheta2}, we introduce a second ansatz and look for a solution $\theta$ of the following form:
\begin{align}
\label{ansatz2}
\theta(t,q,S) = q^\intercal A(t)q + q^\intercal B(t) S + S^\intercal C(t)S + D(t)^\intercal q + E(t)^\intercal S + F(t) \quad \forall (t,q,S) \in [0,T] \times \mathbb R^d \times \mathbb R^d
\end{align}
or equivalently
$$ \theta(t,q,S) = \begin{pmatrix} q \\ S\end{pmatrix}^\intercal P(t) \begin{pmatrix} q \\ S\end{pmatrix} + \begin{pmatrix} D(t) \\ E(t)\end{pmatrix}^\intercal\begin{pmatrix} q \\ S\end{pmatrix} + F(t) \quad \forall (t,q,S) \in [0,T] \times \mathbb R^d \times \mathbb R^d,$$
where $P : [0,T] \rightarrow \mathcal S_{2d}(\mathbb R)$ is defined as
\begin{align}
\label{Pmatdefinition}
P(t) = \begin{pmatrix} A(t) & \frac 12 B(t) \\ \frac 12 B(t)^\intercal & C(t) \end{pmatrix}. 
\end{align}

The interest of this ansatz is stated in the following proposition:
\begin{prop} \label{prop1}
Let $\tau <T$. Assume there exist $A \in C^1 \left([\tau,T], \mathcal S_d(\mathbb R) \right)$, $B \in C^1 \left([\tau,T], \mathcal M_d(\mathbb R) \right)$, $C \in C^1 \left([\tau,T], \mathcal S_d(\mathbb R) \right)$, $D \in C^1 \left([\tau,T], \mathbb R^d \right)$, $E \in C^1 \left([\tau,T], \mathbb R^d \right)$, $F \in C^1 \left([\tau,T], \mathbb R \right)$ satisfying the system of ODEs 
\begin{align}
\label{ODEsystem}
\begin{cases}
A'(t) = \frac{\gamma}{2} (B(t) +I_d) \Sigma (B(t)^\intercal +I_d) - A(t) \eta^{-1} A(t) \\
B'(t) = (B(t)+I_d)R + 2\gamma  (B(t) +I_d) \Sigma C(t) - A(t) \eta^{-1 } B(t)\\
C'(t) = R^\intercal C(t) + C(t)R + 2\gamma C(t)\Sigma C(t) - \frac 1{4} B(t)^\intercal \eta^{-1} B(t) \\
D'(t) = -(B(t)+I_d) R\overline S  + \gamma  (B(t) +I_d)\Sigma E(t) - A(t)\eta^{-1}D(t)\\
E'(t) = -2 C(t)R\overline S + R^\intercal E(t) +2\gamma C(t)\Sigma E(t) - \frac 1{2} B(t)^\intercal \eta^{-1} D(t)\\
F'(t) = -\overline S^\intercal R^\intercal  E(t) - Tr(\Sigma C(t)) +\frac{\gamma}2 E(t)^\intercal \Sigma E(t) - \frac 1{4}D(t)^\intercal \eta^{-1} D(t),
\end{cases}
\end{align}

where $I_d$ denotes the identity matrix in $\mathcal M_d (\mathbb R)$, with terminal conditions
\begin{align}
\label{termcondABCDEF}
A(T) = - \Gamma, \quad B(T) = C(T) = D(T) = E(T) = F(T)=0.
\end{align}

Then the function $\theta$ defined by \eqref{ansatz2} satisfies \eqref{HJBtheta2} on $[\tau,T) \times \mathbb R^d \times \mathbb R^d$ with terminal condition \eqref{termcondtheta1}.
\end{prop}
\begin{proof}
\allowdisplaybreaks

Let us consider $A \in C^1 \left([\tau,T], \mathcal S_d(\mathbb R) \right)$, $B \in C^1 \left([\tau,T], \mathcal M_d(\mathbb R) \right)$, $C \in C^1 \left([\tau,T], \mathcal S_d(\mathbb R) \right)$, $D \in C^1 \left([\tau,T], \mathbb R^d \right)$, $E \in C^1 \left([\tau,T], \mathbb R^d \right)$, $F \in C^1 \left([\tau,T], \mathbb R \right)$ verifying \eqref{ODEsystem} on $[\tau,T)$ with terminal condition \eqref{termcondABCDEF}. Let us consider $\theta: [\tau,T] \times \mathbb R^d \times \mathbb R^d \rightarrow \mathbb R$ defined by \eqref{ansatz2}. Then we obtain for all $(t,q,S) \in [\tau,T) \times \mathbb R^d \times \mathbb R^d$:
\begin{align*}
& \partial_t \theta(t,q,S) + \frac 14 \nabla_q \theta(t,q,S)^\intercal \eta^{-1}\nabla_q \theta(t,q,S) +\frac 12 \textrm{Tr} \left(  \Sigma D^2_{SS}\theta(t,q,S) \right)\\
& - \frac{\gamma}{2}  (q + \nabla_S \theta(t,q,S))^\intercal \Sigma (q + \nabla_S \theta(t,q,S)) +  (\overline S-S)^\intercal R^\intercal  (\nabla_S \theta(t,q,S)+q), \nonumber \\
=\quad & q^\intercal A'(t) q + q^\intercal B'(t) S + S^\intercal C'(t)S + D'(t)^\intercal q + E'(t)^\intercal S + F'(t)\nonumber\\
& + q^\intercal A(t) \eta^{-1} A(t) q  + q^\intercal A(t)  \eta^{-1} B(t) S + \frac 14 S^\intercal B(t)^\intercal \eta^{-1} B(t) S \\
& + D(t)^\intercal \eta^{-1} A(t) q + \frac 12 \left(D(t) \right)^\intercal \eta^{-1} B(t) S + \frac 14 D(t)^\intercal \eta^{-1} D(t)\\
& + Tr(\Sigma C(t)) - \frac{\gamma}{2}  \left(q + B(t)^\intercal q + 2C(t)S + E(t) \right)^\intercal \Sigma \left(q + B(t)^\intercal q + 2C(t)S + E(t) \right)  \\
& + \overline S^\intercal R^\intercal q + \overline S ^\intercal R^\intercal \left(B(t)^\intercal q + 2C(t)S + E(t) \right) - S^\intercal R^\intercal q - S^\intercal R^\intercal \left(B(t)^\intercal q + 2C(t)S + E(t) \right)\\
=\quad & q^\intercal \left(A'(t) - \frac{\gamma}2 (B(t) + I_d) \Sigma (B(t)^\intercal + I_d) + \frac 14 \left(2A(t) \right)\eta^{-1} \left(2A(t) \right)\right)q\\
& + q^\intercal \left(B'(t) - (I_d + B(t))R - 2\gamma (B(t) + I_d) \Sigma C(t) +  A(t) \eta^{-1} B(t) \right) S \\
& + S^\intercal \left(C'(t) - R^\intercal C(t) - C(t)R - 2\gamma C(t)\Sigma C(t) + \frac 1{4} B(t)^\intercal \eta^{-1} B(t) \right) S\\
& + \left(D'(t) +(B(t)+I_d) R\overline S  - \gamma  (B(t) +I_d)\Sigma E(t) + A(t)\eta^{-1}D(t) \right)^\intercal q\\
& + \left(E'(t) +2 C(t)R\overline S - R^\intercal E(t) -2\gamma C(t)\Sigma E(t) + \frac 1{2} B(t)^\intercal \eta^{-1} D(t) \right) ^\intercal S\\
& + \left(F'(t) +\overline S^\intercal R^\intercal  E(t) + Tr(\Sigma C(t)) -\frac{\gamma}2 E(t)^\intercal \Sigma E(t) + \frac 1{4} D(t)^\intercal \eta^{-1}D(t) \right)\\
= \quad & 0.
\end{align*}
As it is straightforward to verify that $\theta$ satisfies the terminal condition \eqref{termcondtheta1}, the result is proved.
\end{proof}


\begin{rem} \label{remark1} Two remarks can be made on the system of ODEs \eqref{ODEsystem}:
\begin{itemize}
    \item This system of ODEs can clearly be decomposed into three groups of equations: the first three ODEs for $A,$ $B$ and $C$ are independent of the others and can be solved as a first step; once we know $A,B,$ and $C$ we can solve the linear ODEs for $D$ and $E$, and finally $F$ can be obtained with a simple integration;
    \item When $R=0$ (i.e. in the case where the prices $S$ of the $d$ assets are correlated arithmetic Brownian motions), there is a trivial solution to the last five equations which is $B=C=D=E=F=0.$ The function $A$ can then be found using classical techniques (as shown in Appendix \ref{annex1}).
\end{itemize}
\end{rem}

It is noteworthy that the first system, i.e.
\begin{align}
\label{ODEsystemABC}
\begin{cases}
A'(t) = \frac{\gamma}{2} (B(t) +I_d) \Sigma (B(t)^\intercal +I_d) - A(t) \eta^{-1} A(t) \\
B'(t) = (B(t)+I_d)R + 2\gamma  (B(t) +I_d) \Sigma C(t) - A(t) \eta^{-1 } B(t)\\
C'(t) = R^\intercal C(t) + C(t)R + 2\gamma C(t)\Sigma C(t) - \frac 1{4} B(t)^\intercal \eta^{-1} B(t)
\end{cases}
\end{align}
boils down to the following Matrix Riccati ODE in $P = \begin{pmatrix} A & \frac 12 B \\ \frac 12 B & C \end{pmatrix} $:
\begin{align}
\label{PODE}
P'(t) = Q + Y^\intercal P(t) +  P(t) Y +  P(t)UP(t),
\end{align}
where 
\begin{align*}
Q =\ & \frac 12 \begin{pmatrix} \gamma \Sigma & R\\ R^\intercal & 0 \end{pmatrix} \in \mathcal S_{2d}(\mathbb R), & 
Y =\ & \begin{pmatrix} 0 & 0\\ \gamma \Sigma & R\end{pmatrix} \in \mathcal M_{2d}(\mathbb R), & 
U =\ & \begin{pmatrix} -\eta^{-1} & 0 \\ 0 & 2\gamma \Sigma \end{pmatrix} \in \mathcal S_{2d}(\mathbb R),
\end{align*}
and the terminal condition writes\vspace{-0.2cm}
\begin{align}
\label{termcondP}
P(T) = \begin{pmatrix} - \Gamma & 0 \\ 0 & 0 \end{pmatrix} \in \mathcal S_{2d}(\mathbb R).
\end{align}

When compared to the Matrix Riccati ODEs arising in the linear-quadratic optimal control literature, the distinctive aspect of our equation is that the matrix $U$ characterizing the quadratic term in the Riccati equation has both positive and negative eigenvalues. In particular, we cannot rely on existing results (see for instance Theorem 3.5 of \cite{freiling2002survey}) to prove that there exists a solution to \eqref{PODE} with terminal condition \eqref{termcondP}. In this paper, we address the existence of a solution by using \textit{a priori} estimates for the value function.\footnote{Surprisingly maybe, the inequalities we derive do not seem to derive in a direct manner from a purely analytic argument such as the classical comparison principle for matrix Riccati equations -- see Theorem 3.4 of \cite{freiling2002survey}.} \\

Regarding the set of equations \eqref{ODEsystemABC}, there exists a unique local solution by Cauchy-Lipschitz theorem. In the following section, we therefore first state a verification theorem that solves the problem on an interval $[\tau, T]$, and use that very result to address global existence and uniqueness of a solution on $[0, T]$.

\subsection{Main mathematical results}

\begin{thm} \label{verification}
Let $\tau<T$. Let $A \in C^1 \left([\tau,T], \mathcal S_d(\mathbb R) \right)$, $B \in C^1 \left([\tau,T], \mathcal M_d(\mathbb R) \right)$, $C \in C^1 \left([\tau,T], \mathcal S_d(\mathbb R) \right)$, $D \in C^1 \left([\tau,T], \mathbb R^d \right)$, $E \in C^1 \left([\tau,T], \mathbb R^d \right)$, $F \in C^1 \left([\tau,T], \mathbb R \right)$ be a solution to the system \eqref{ODEsystem} on $[\tau, T)$ with terminal condition \eqref{termcondABCDEF}, and consider the function $\theta$ defined by \eqref{ansatz2} and the associated function $w$ defined by \eqref{ansatz1}.\\

Then for all $(t,x,q,S) \in [\tau,T] \times \mathbb R \times \mathbb R^d \times \mathbb R^d$ and $v = (v_s)_{s \in [t,T]} \in \mathcal A_t$, we have 
\begin{align}
\label{subopt}
\mathbb E \left[ -e^{-\gamma \left(X^{t,x,S,v}_T + \left( q^{t,q,v}_T\right)^\intercal S^{t,S}_T - \ell(q^{t,q,v}_T) \right)} \right] \le w(t,x,q,S).
\end{align}
Moreover, equality is obtained in \eqref{subopt} by taking the optimal control $(v^*_s)_{s \in [t,T]} \in \mathcal A_t$ given by the closed-loop feedback formula
\begin{align}
\label{optcontrol}
v^*_s = \frac 12 \eta^{-1} \left( 2A(s)q^{t,q,v}_s + B(s)S^{t,S}_s + D(s)  \right).
\end{align}
In particular, $w=u$ on $[\tau,T] \times \mathbb R \times \mathbb R^d \times \mathbb R^d.$
\end{thm}

\begin{proof}
Let $t \in [\tau, T)$, we first prove that $(v^*_s)_{s \in [t,T]} \in \mathcal A_t$ (i.e., $(v^*_s)_{s \in [t,T]}$ is well-defined and admissible).
Let us consider the Cauchy initial value problem
\begin{align*}
	\forall s \in [t,T],\quad \frac {d\tilde{q}_s}{ds}= \frac 12 \eta^{-1} \left( 2 A(s) \tilde{q}_s + B(s)S^{t,S}_s + D(s) \right), \qquad \tilde{q}_t=q. \notag
\end{align*}
The unique solution of that Cauchy problem writes
\begin{align*}
    \tilde{q}_s= \exp \left( \int_{t}^{s} \phi(\varrho) d\varrho{} \right) \left( q + \int_{t}^{s} \psi\left( \varrho, S^{t,S}_\varrho \right)\exp \left( - \int_{t}^{\varrho} \phi(\varsigma) d\varsigma{} \right)d\varrho \right),
\end{align*}
where $\phi$ and $\psi$ are defined by
\begin{align*}
    \phi& : s \in[t,T] \mapsto  \eta^{-1} A(s), \notag \\
   \psi& :  (s, S)\in [t,T] \times \mathbb R^d \mapsto \frac 12 \eta^{-1} \left( B(s) S + D(s)  \right). \notag
\end{align*}
Then $v^*$ can be written as
\begin{align*}
    v^{*}_s =  \frac {d\tilde{q}_s}{ds} = \phi(s) \exp \left( \int_{t}^{s} \phi(\varrho) d\varrho \right) \left( q + \int_{t}^{s} \psi\left( \varrho, S^{t,S}_\varrho \right)\exp \left( - \int_{t}^{\varrho} \phi(\varsigma) d\varsigma{} \right)d\varrho \right) + \psi\left(s,S^{t,S}_s\right).
\end{align*}
We see from the definition of $\phi$ and the affine form of $\psi$ in $S$ that $v^*$ satisfies a linear growth condition, and is therefore in $\mathcal A_t$.\\

Let us consider $(t,x,q,S) \in [\tau,T] \times \mathbb R \times \mathbb R^d \times \mathbb R^d$ and $v = (v_s)_{s \in [t,T]} \in \mathcal A_t$. We now prove that 
\begin{align*}
    \mathbb E \left[w\left(T,X^{t,x,S,v}_T,q^{t,q,v}_T,S^{t,S}_T\right) \right] \leqslant w(t,x,q,S).
\end{align*}
We use the following notations for readability:
\begin{align*}
    \forall s \in [t,T], & \hspace{0.5cm} w(s,X^{t,x,S,v}_s,q^{t,q,v}_s,S^{t,S}_s)=w^{t,x,q,S,v}_s, \notag \\
    \forall s \in [t,T], & \hspace{0.5cm} \theta (s,q^{t,q,v}_s,S^{t,S}_s)=\theta^{t,q,S,v}_s. \notag
\end{align*}

By Itô's formula, we have $\forall s \in [\tau,T]$
\begin{align*}
    dw^{t,x,q,S,v}_s = \mathcal{L}^{v}w^{t,x,q,S,v}_s ds + \left(\nabla_{S}w^{t,x,q,S,v}_s\right)^{\intercal} VdW_s,
\end{align*}
where
\begin{align*}
    \mathcal{L}^{v}w^{t,x,q,S,v}_s=\ &\partial_t w^{t,x,q,S,v}_s - (v^\intercal S+v^\intercal \eta v)\partial_x w^{t,x,q,S,v}_s + v^\intercal \nabla_q w^{t,x,q,S,v}_s \nonumber\\
    & + (\overline S-S)^\intercal R^\intercal  \nabla_S w^{t,x,q,S,v}_s + \frac 12 \textrm{Tr}\left( \Sigma D^2_{SS} w^{t,x,q,S,v}_s \right).
\end{align*}
From \eqref{ansatz1} and \eqref{ansatz2} we have
\begin{align*}
    \nabla_S w^{t,x,q,S,v}_s&=-\gamma w^{t,x,q,S,v}_s \left( q^{t,q,v}_s + \nabla_S \theta ^{t,q,S,v}_s \right) \nonumber\\
    & = -\gamma w^{t,x,q,S,v}_s \left(q^{t,q,v}_s + B(s)^\intercal q^{t,q,v}_s + 2 C(s)S^{t,S}_s + E(s) \right).
\end{align*}
We define $\forall s \in [t, T]$,
\begin{align*}
    \kappa^{q,S,v}_s &= -\gamma \left(q^{t,q,v}_s + B(s)^\intercal q^{t,q,v}_s + 2 C(s)S^{t,S}_s + E(s) \right), \\
    \xi^{q,S,v}_{t,s} &= \exp \left( \int_{t}^{s} {\kappa^{q,S,v}_\varrho}^\intercal VdW_\varrho - \frac12 \int_{t}^{s} {\kappa^{q,S,v}_\varrho}^\intercal \Sigma \kappa^{q,S,v}_\varrho d\varrho \right).
\end{align*}
We then have 
\begin{align*}
    d \left( w^{t,x,q,S,v}_s \left( \xi^{q,S,v}_{t,s} \right)^{-1} \right) = \left( \xi^{q,S,v}_{t,s} \right)^{-1} \mathcal{L}^{v}w^{t,x,q,S,v}_s ds.
\end{align*}
By definition of $w$, $\mathcal{L}^{v}w^{t,x,q,S,v}_s \leqslant 0$. Moreover, equality holds for the control reaching the supremum in \eqref{LegendreFenchel}. This supremum is reached for the unique value
\begin{align*}
    v_s &= \frac12 \eta^{-1} \nabla_q \theta ^{t,q,S,v}_s    \\
    &=\frac 12 \eta^{-1} \left( 2A(s) q^{t,q,v}_s + B(s) S^{t,S}_s + D(s)  \right),
\end{align*}
which corresponds to the case $(v_s)_{s \in [t,T]} = (v^*_s)_{s \in [t,T]}$.\\
\\ \\
As a consequence, $\left( w^{t,x,q,S,v}_s \left( \xi^{q,S,v}_{t,s} \right)^{-1} \right)_{s \in [t,T]}$ is nonincreasing and therefore
\begin{align*}
    w\left(T,X^{t,x,S,v}_T,q^{t,q,v}_T,S^{t,S}_T\right) \leqslant w(t,x,q,S) \xi^{q,S,v}_{t,T},
\end{align*}
with equality when $(v_s)_{s \in [t,T]} = (v^*_s)_{s \in [t,T]}$.
\\
\\
Taking expectations we get 
\begin{align*}
    \mathbb E \left[w\left(T,X^{t,x,S,v}_T,q^{t,q,v}_T,S^{t,S}_T\right) \right] \leqslant w(t,x,q,S)  \mathbb E \left[\xi^{q,S,v}_{t,T} \right].
\end{align*}
We proceed to prove that $\mathbb E \left[\xi^{q,S,v}_{t,T} \right]$ is equal to 1. To do so, we use that $\xi^{q,S,v}_{t,t} = 1$ and prove that $(\xi^{q,S,v}_{t,s})_{s \in [t,T]}$ is a martingale under $\left(\mathbb P; \mathbb F = (\mathcal F_s)_{s \in [t,T]} \right)$. \\

We know that $(q^{t,q,v}_s)_{s \in [t,T]}$ satisfies a linear growth condition with respect to $(S^{t,S}_s)_{s \in [t,T]}$ since $v$ is an admissible control. Given the form of $\kappa$, there exists a constant $C$ such that, almost surely,
\begin{align*}
    \underset{s \in [t,T]}{\sup} {\parallel \kappa^{q,S,v}_s \parallel}^2 \leqslant C \left(1 + \underset{s \in [t,T]}{\sup} {\parallel W_s - W_t \parallel}^2  \right).
\end{align*}
By using classical properties of the Brownian motion, we prove that
\begin{align*}
    \exists \epsilon >0, \forall s \in [t,T], \hspace{0.3cm} \mathbb E \left[ \exp \left( \frac12 \int_{s}^{ \left(s+\epsilon \right) \wedge T}  \left(\kappa^{q,S,v}_\varrho \right)^\intercal \Sigma \kappa^{q,S,v}_\varrho d\varrho\right) \right] < +\infty.
\end{align*}
Using a classical trick due to Bene\v{s} (see \cite{karatzas2014brownian}, Chapter 5), we see that $(\xi^{q,S,v}_{t,s})_{s \in [t,T]}$ is a martingale under $\left(\mathbb P; \mathbb F = (\mathcal F_s)_{s \in [t,T]} \right)$. 
\\
\\
We obtain 
\begin{align*}
    \mathbb E \left[w\left(T,X^{t,x,S,v}_T,q^{t,q,v}_T,S^{t,S}_T\right) \right] \leqslant w(t,x,q,S),
\end{align*}
with equality when $(v_s)_{s \in [t,T]} = (v^*_s)_{s \in [t,T]}$.
\\
\\
We conclude that 
\begin{align*}
    u\left( t,x,q,S\right) &= \underset{(v_s)_{s \in [t,T]} \in \mathcal A_t}{\sup} \mathbb E \left[- \exp \left(-\gamma \left(X^{t,x,S,v}_T + \left( q^{t,q,v}_T\right)^\intercal S^{t,S}_T - \ell\left(q^{t,q,v}_T\right) \right) \right) \right] \\
    &= \mathbb E \left[- \exp \left(-\gamma \left(X^{t,x,S,v^{*}}_T + \left( q^{t,q,v^{*}}_T\right)^\intercal S^{t,S}_T - \ell\left(q^{t,q,v^{*}}_T\right) \right) \right) \right] \\
    &=w \left(t,x,q,S \right).
\end{align*}

\end{proof}

We will next proceed to prove existence and uniqueness of a solution to the system of ODEs \eqref{ODEsystem} on $[0,T]$ with terminal condition \eqref{termcondABCDEF}, or equivalently to \eqref{PODE} with terminal condition \eqref{termcondP}.\footnote{The result in fact holds on $(-\infty, T]$ as the initial time plays no role.}

\begin{thm} \label{existence}
There exists a unique solution $A \in C^1 \left([0,T], \mathcal S_d(\mathbb R) \right)$, $B \in C^1 \left([0,T], \mathcal M_d(\mathbb R) \right)$, $C \in C^1 \left([0,T], \mathcal S_d(\mathbb R) \right)$, $D \in C^1 \left([0,T], \mathbb R^d \right)$, $E \in C^1 \left([0,T], \mathbb R^d \right)$, $F \in C^1 \left([0,T], \mathbb R \right)$ to the system of ODEs \eqref{ODEsystem} on $[0,T]$ with terminal condition~\eqref{termcondABCDEF}.
\end{thm}

\begin{proof}
To prove Theorem \ref{existence}, it is enough, as explained in Remark \ref{remark1}, to show existence and uniqueness for $A \in C^1 \left([0,T], \mathcal S_d(\mathbb R) \right)$, $B \in C^1 \left([0,T], \mathcal M_d(\mathbb R) \right)$, and $C \in C^1 \left([0,T], \mathcal S_d(\mathbb R) \right)$, or equivalently, existence and uniqueness on $[0,T]$ of a solution $P \in C^1 \left([0,T], \mathcal S_{2d}(\mathbb R) \right) $ to \eqref{PODE} with terminal condition \eqref{termcondP}.\\

By Cauchy-Lipschitz theorem, there exists a unique maximal solution\footnote{The fact that $A$ and $C$ are symmetric is itself a consequence of Cauchy-Lipschitz theorem since $(A,B,C)$ and $(A^\intercal, B, C^\intercal)$ are solution of the same Cauchy problem.} $(A,B,C)$ to the system of ODEs \eqref{ODEsystemABC} with terminal condition \eqref{termcondABCDEF} defined on an open interval $(t_{\text{min}}, t_{\text{max}}) \ni T$, and by Theorem \ref{verification}, the associated function $w$ defined by \eqref{ansatz1} corresponds to the value function of the problem restricted to $[\tau, T]$ for all $\tau \in (t_{\text{min}}, T)$.\\

To prove our result, we need to show that $t_{\text{min}} = -\infty$. For that purpose, our strategy consists in proving that the matrix $P(t) = \begin{pmatrix} A(t) & \frac 12 B(t) \\ \frac 12 B(t)^\intercal & C(t) \end{pmatrix}$ cannot blow up in finite time. This is proved by, first, finding (thanks to the control problem) lower and upper bounds for the function $\theta$ which are, like $\theta$, polynomials of degree at most $2$ in $(q,S)$, and then converting these bounds into bounds for $P(t)$ in the sense of the natural order on symmetric matrices.\footnote{For $\underline M,\overline M  \in \mathcal S_d(\mathbb R)$, $\underline M \leq \overline M $ if and only if $\overline M - \underline M  \in \mathcal S^{+}_d(\mathbb R)$.}\\ 

By contradiction, let us assume that $t_{\text{min}} \in (-\infty, T)$ and let $\tau \in (t_{\text{min}}, T)$.\\ 
 
Let $(t,x,q,S) \in [\tau,T] \times \mathbb R \times \mathbb R^d \times \mathbb R^d$ and let us consider the sub-optimal strategy $v = (0)_{s \in [t,T]} \in \mathcal A_t$ for which $\forall s \in [t,T], q^{t,q,v}_s=q$ and
\begin{align}
\label{avantLaplace}
     \mathbb E \left[ - \exp \left(-\gamma \left( X_T^{t,x,S,v} + (q^{t,q,v}_T)^\intercal S^{t,S}_T - \ell(q^{t,q,v}_T) \right)\right) \right] = \mathbb E\left[ - \exp \left(-\gamma \left(x + q^\intercal S + q^\intercal \left(S^{t,S}_T-S\right) - q^\intercal \Gamma q  \right) \right) \right].
\end{align}

Since $(S^{t,S}_s)_{s \in [t,T]}$ follows multivariate Ornstein-Uhlenbeck dynamics, we know that \vspace{-0.2cm}
\begin{align}
\label{MOUint}
    S^{t,S}_T - S = \left(I - e^{-R\left(T-t\right)} \right) \left(\overline{S} - S\right) + \int_{t}^T e^{-R\left(T-s\right)}VdW_s.
\end{align}\vspace{-0.18cm}

Then $S^{t,S}_T - S \sim \mathcal{N} \left(\left(I - e^{-R\left(T-t\right)} \right) \left(\overline{S} - S\right),\Sigma_t\right)$, where the covariance matrix is defined by\vspace{-0.2cm} $$\Sigma_{t} = \int_{t}^T e^{-R\left(T-s\right)} \Sigma e^{-R^\intercal\left(T-s \right)} ds.\vspace{-0.2cm}$$

Then,\vspace{-0.2cm}
\begin{align*}
&\mathbb E \left[ - \exp \left(-\gamma \left( X_T^{t,x,S,v} + (q^{t,q,v}_T)^\intercal S^{t,S}_T - \ell(q^{t,q,v}_T) \right)\right) \right]\nonumber\\
=\quad& -\exp \left(-\gamma(x + q^\intercal S)\right) \exp \left( -\gamma \left(q^\intercal \left(I - e^{-R\left(T-t\right)} \right) \left(\overline{S} - S\right) - q^\intercal \Gamma q - \frac{1}{2} \gamma q^\intercal \Sigma_t q \right) \right).
\end{align*}

Since the strategy is sub-optimal, if we consider $\theta$ defined as in \eqref{ansatz2}, we have by Theorem \ref{verification}
\begin{align}
\label{lowerBound}
    -\exp \left(-\gamma \left(x+q^\intercal S + \theta(t,q,S) \right) \right) & \geq -\exp \left(-\gamma(x + q^\intercal S)\right) \exp \left( -\gamma \left(q^\intercal \left(I - e^{-R\left(T-t\right)} \right) \left(\overline{S} - S\right) - q^\intercal \Gamma q - \frac{1}{2} \gamma q^\intercal \Sigma_t q \right) \right).
\end{align}

We conclude that for all $(t,q,S) \in [\tau,T] \times \mathbb R^d \times \mathbb R^d$,
\begin{align*}
&\theta(t,q,S) = \begin{pmatrix} q \\ S \end{pmatrix}^\intercal P(t) \begin{pmatrix} q \\ S \end{pmatrix} + \begin{pmatrix} D(t) \\ E(t)\end{pmatrix}^\intercal\begin{pmatrix} q \\ S\end{pmatrix} + F(t)\\
\geq\quad & \begin{pmatrix} q \\ S \end{pmatrix}^\intercal \begin{pmatrix} - \frac{\gamma}{2} \Sigma_t - \Gamma & - \frac12  \left(I - e^{-R\left(T-t\right)} \right) \\ - \frac12  \left(I - e^{-R^\intercal \left(T-t\right)} \right) & 0 \end{pmatrix}\begin{pmatrix} q \\ S \end{pmatrix} + \overline{S}^\intercal \left(I - e^{-R^\intercal\left(T-t\right)} \right) q.\\
\end{align*}

We therefore necessarily have, for the natural order on symmetric matrices,
$$\forall t \in [\tau,T], \quad P(t) \geq \begin{pmatrix} - \frac{\gamma}{2} \Sigma_t - \Gamma & - \frac12  \left(I - e^{-R\left(T-t\right)} \right) \\ - \frac12  \left(I - e^{-R^\intercal \left(T-t\right)} \right) & 0 \end{pmatrix}.$$

Now, for $(t,x,q,S) \in [\tau,T] \times \mathbb R \times \mathbb R^d \times \mathbb R^d$, we have
\begin{align*}
   & \underset{v \in \mathcal A_t}{\sup} \mathbb E \left[ - \exp \left(-\gamma \left( X^{t,x,S,v}_T + (q^{t,q,v}_T)^\intercal S^{t,S}_T - (q^{t,q,v}_T)^\intercal \Gamma q^{t,q,v}_T \right) \right) \right]\\
   =  &  \underset{v \in \mathcal A_t}{\sup} \mathbb E \left[ - \exp \left(-\gamma \left(x + q^{\intercal} S + \int_t^{T} (q^{t,q,v}_s)^\intercal dS_s - \int_t^{T} L(v_s)ds - (q^{t,q,v}_T)^\intercal \Gamma q^{t,q,v}_T \right) \right) \right] \nonumber \\
    \leq & \exp \left(-\gamma \left( x + q^{\intercal} S\right) \right) \underset{v \in \mathcal A_t}{\sup} \mathbb E \left[ - \exp \left(-\gamma \left(\int_t^{T} (q^{t,q,v}_s)^\intercal dS_s \right) \right) \right],
\end{align*}

If $(v_s)_{s\in [t,T]} \in \mathcal A_t$, it is straightforward to see that the process $(q^{t,q,v}_s)_{s \in [t,T]}$ is in the space of admissible controls $\mathcal A_t^{Merton}$ defined by \eqref{Amerton} in Appendix \ref{annexMerton} (in which we study a Merton problem that can be regarded as a limit case of ours when the execution costs and terminal costs vanish). Therefore,
\begin{align}
    \label{ineqUpperBound}
    &\underset{v \in \mathcal A_t}{\sup} \mathbb E \left[ - \exp \left(-\gamma \left(  X^{t,x,S,v}_T + (q^{t,q,v}_T)^\intercal S^{t,S}_T - (q^{t,q,v}_T)^\intercal \Gamma q^{t,q,v}_T  \right) \right) \right] \nonumber\\
    \leq & \exp \left(-\gamma \left( x + q^{\intercal} S\right) \right) \underset{(q_s)_{s \in [t,T]} \in \mathcal A_t^{Merton}}{\sup} \mathbb E \left[ - \exp \left(-\gamma \left(\int_t^{T} q^{\intercal}_s dS_s \right) \right) \right].
\end{align}

As shown in Appendix \ref{annexMerton}, inequality \eqref{ineqUpperBound} writes
\begin{align*}
    -\exp \left(-\gamma \left(x+q^\intercal S + \theta(t,q,S) \right) \right) \leq -\exp \left(-\gamma \left(x+q^\intercal S + \hat{\theta}(t,S) \right) \right),
\end{align*}
where $\hat \theta(t,S) = S^\intercal \hat C(t) S + \hat E(t)^\intercal S + \hat F(t)$ with $\hat C \in C^1 \left([\tau,T], \mathcal S_d(\mathbb R) \right)$, $\hat E \in C^1 \left([\tau,T], \mathbb R^d \right)$, $\hat F \in C^1 \left([\tau,T], \mathbb R \right)$ defined by
$$
\begin{cases}
\hat C(t) = \frac 1{2\gamma} (T-t)  R^\intercal \Sigma^{-1} R,\\
\hat E(t) = -\frac 1{\gamma} (T-t)  R^\intercal \Sigma^{-1} R \overline S,\\
\hat F(t) = \frac 1{4\gamma} (T-t)^2 \text{Tr}\left( R^\intercal \Sigma^{-1} R \Sigma \right) + \frac{1}{2\gamma} (T-t)  \overline S^\intercal R^\intercal \Sigma^{-1} R.
\end{cases}
$$

We conclude that for all $(t,q,S) \in [\tau,T] \times \mathbb R^d \times \mathbb R^d$,
\begin{align*}
&\theta(t,q,S) = \begin{pmatrix} q \\ S \end{pmatrix}^\intercal P(t) \begin{pmatrix} q \\ S \end{pmatrix} + \begin{pmatrix} D(t) \\ E(t)\end{pmatrix}^\intercal\begin{pmatrix} q \\ S\end{pmatrix} + F(t)\\
\leq\quad & \begin{pmatrix} q \\ S \end{pmatrix}^\intercal \begin{pmatrix} 0 & 0 \\0 & \hat C(t)\end{pmatrix}\begin{pmatrix} q \\ S \end{pmatrix} + \begin{pmatrix} 0 \\ \hat E(t)\end{pmatrix}^\intercal\begin{pmatrix} q \\ S\end{pmatrix} + \hat F(t).\\
\end{align*}
Therefore,
$$\forall t \in [\tau, T], \quad P\left(t\right) \leq \begin{pmatrix} 0 & 0 \\0 & \hat C(t) \end{pmatrix} = \begin{pmatrix} 0 & 0 \\0 & \frac 1{2\gamma}(T-t)  R^\intercal \Sigma^{-1} R \end{pmatrix}.$$

We have therefore $\forall \tau \in (t_\text{min}, T)$, $\forall t \in [\tau,T]$:
$$\begin{pmatrix} - \frac{\gamma}{2} \Sigma_t - \Gamma & - \frac12  \left(I - e^{-R\left(T-t\right)} \right) \\ - \frac12  \left(I - e^{-R^\intercal \left(T-t\right)} \right) & 0 \end{pmatrix} \leq P\left(t\right) \leq  \begin{pmatrix} 0 & 0 \\0 & \frac 1{2\gamma} (T-t)  R^\intercal \Sigma^{-1} R \end{pmatrix}.$$

As $t_\text{min}$ is supposed to be finite, there exists $\underline M,\overline M \in \mathcal S_d(\mathbb R)$ with $\underline M \le \overline M$ such that $\forall t \in [t_\text{min},T]$, $P(t)$ stays in the compact set 
$ \{ M \in \mathcal S_d(\mathbb R)\,  \left|  \underline M\leq M \leq  \overline M \right. \}$. This contradicts the maximality of the solution, hence $t_\text{min} = -\infty$.\\ 
\end{proof}

Theorem \ref{existence} implies that Theorem \ref{verification} can be applied with $\tau = 0$. In particular, our optimal execution problem is solved and the optimal strategy is given by the closed-loop feedback control   \eqref{optcontrol}. In the next section, we illustrate our results with simulations of prices and numerical approximations of the optimal strategies.

\section{Numerical results}

In this section, we present several applications of our results. We first exemplify the use of the optimal strategy derived in the above section by a trader wishing to unwind a single-asset portfolio. In particular, we show that the optimal liquidation strategy in our model with mean reversion is really different from that derived in the Almgren-Chriss model. We also demonstrate the usefulness of our results for statistical arbitrage purposes in the one-asset case. The one-asset examples are based on data from the foreign exchange (FX) market. We then illustrate our results in the multi-asset case by considering a pair of two cointegrated French stocks. We start with a two-asset portfolio liquidation problem and compare the optimal liquidation strategy in our model with that obtained in a multi-asset Almgren-Chriss model. We then illustrate the use of our results for statistical arbitrage purposes (a pair trading strategy in our case).\footnote{Throughout this section, optimal trading strategies are computed by approximating the solution of the Riccati ODEs using implicit Euler schemes.}\\

In the Almgren-Chriss model used for carrying out comparisons, the price dynamics is of the form $dS_t = V_{\textrm{AC}} dW_t$, where $V_{\textrm{AC}} \in \mathcal M_{d,k}(\mathbb R)$, i.e. a simple Bachelier dynamics (with correlations). This dynamics differs from that of the OU model we use throughout the paper (i.e. $dS_t = R(\overline S - S_t) dt + V dW_t$) when $R \neq 0$. In particular, if prices exhibit mean reversion or a cointegrated behavior, as is the case in our examples, the classical Almgren-Chriss model shall not properly take the true multivariate dynamics of prices into account, with sometimes important consequences in terms of risk management.

\subsection{Single-asset case}

In order to illustrate the use of the optimal strategies we derived in the above section, let us start with a single-asset case. For that purpose, we use data from the FX market, in which asset prices often exhibit mean reversion. More precisely, we consider a FX futures contract (hereafter CDU1) on the currency pair Canadian Dollar (CAD) / US~Dollar~(USD) that is exchanged on the Chicago Mercantile Exchange. The contract specifications are given in Table \ref{table:CDU1}. 

\begin{table}[H]
\begin{center}
\begin{tabular}{c  c} 
 \hline
Underlying asset & Canadian Dollar \\ [0.5ex] 
Quotation currency & US Dollar \\ [0.5ex] 
Contract size & CAD 100000 \\ [0.5ex]
Expiry date & September 14, 2021 \\ 
\hline 
\end{tabular}
\end{center}
\caption {CDU1 contract specifications}
\label{table:CDU1}
\end{table}

We plot in Figure \ref{Asset_prices_CAD} the mid-price of CDU1,\footnote{CDU1 is usually quoted in USD cents per CAD. However, to use our model, the price must take account of the contract size and be the contract value in USD.} sampled every 60 seconds during the regular trading hours ($02$:$00$-$16$:$00$~Central Time),\footnote{Although CDU1 is quoted continuously with a 60-minute break each day beginning at $16$:$00$ CT, we only consider the trading hours between $02$:$00$ and $16$:$00$ CT because the contract is only liquid during these hours corresponding to European and American market activity.} over the three following trading days: August 11, August 12, and August 13, 2021. \\

\begin{figure}[!h]\centering
\includegraphics[width=0.88\textwidth]{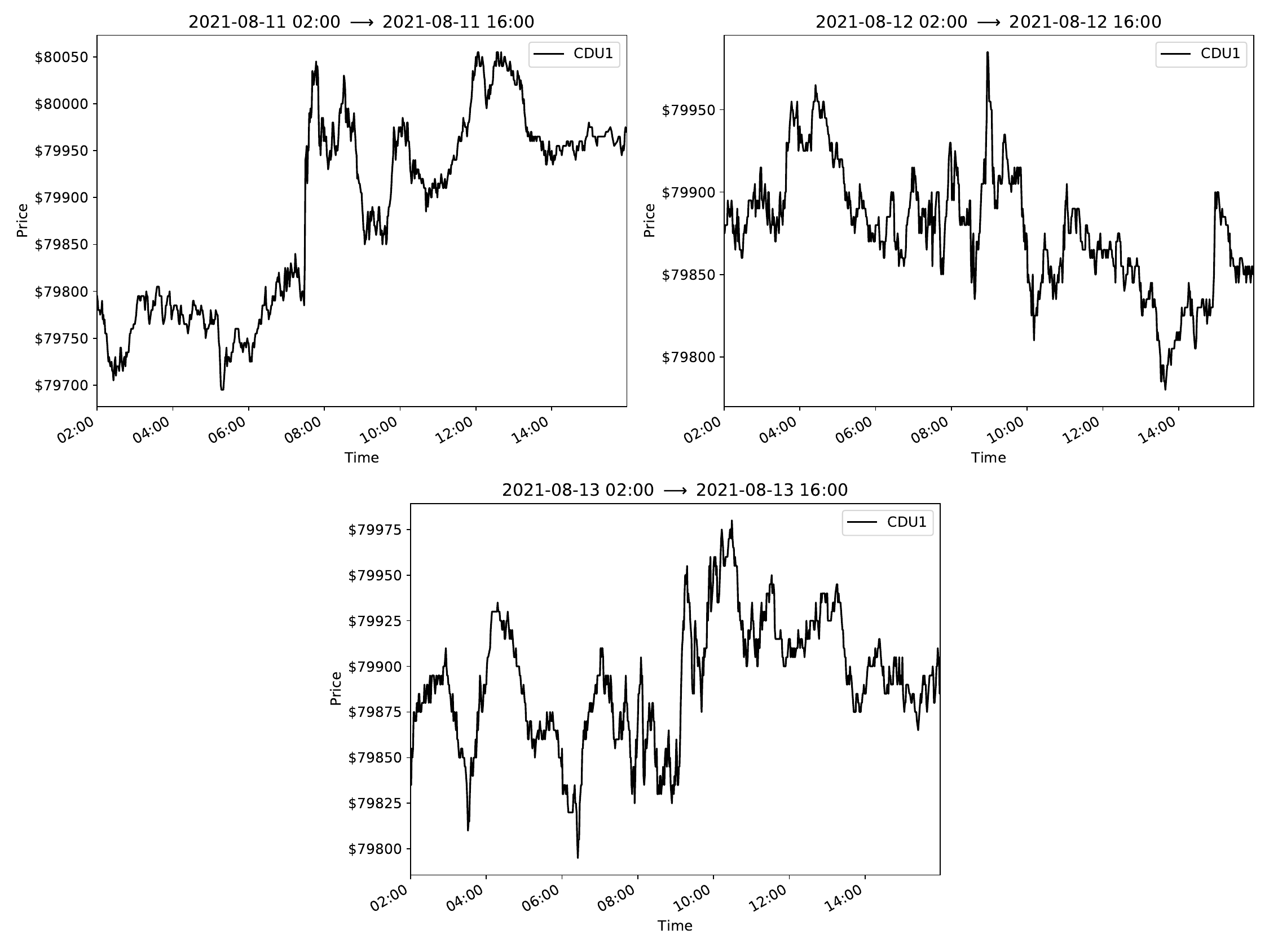}\\
\caption{Mid-price of CDU1 sampled every 60 seconds during the regular trading hours ($02$:$00$-$16$:$00$ CT). Top left: August~11, 2021. Top right: August 12, 2021. Bottom: August 13, 2021.}\label{Asset_prices_CAD}
\end{figure}

\subsubsection{Liquidation problem}

We consider the case of a trader wishing to unwind a long position in 2250 contracts\footnote{This represents roughly $5\%$ of the average daily traded volume over the period considered in this example.} during the third day, i.e. on August 13, 2021.\\

To exemplify the use of our strategies, we first estimate Ornstein-Uhlenbeck parameters using prices from the two preceding trading days: August 11 and August 12, 2021. Coefficients are classically estimated using least squares regression.\footnote{A time discretization of an Ornstein-Uhlenbeck model gives rise to an Auto-Regressive model of order $1$, or $\textrm{AR}(1)$. The parameters of an $\textrm{AR}(1)$ model are classically estimated by using least squares regression. Conversion of $\textrm{AR}(1)$ coefficients into their continuous-time counterparts is straightforward.} In order to set the value of the execution cost / temporary market impact parameter $\eta$, we use a similar argument as in \cite{almgren2001optimal}: we suppose that the additional cost incurred per contract when trading a given volume is proportional to the participation rate to the market. More precisely, for each percent of participation rate (in practice we consider a flat volume curve that matches the average daily volume), a cost corresponding to half the bid-ask spread\footnote{The average bid-ask spread is close to the tick value equal to $\$5$ per contract.} is incurred. Rounding values, this results in setting $\eta = 5 \cdot 10^{-3} \ \$ \cdot \textrm{day}$. For the terminal penalty parameter $\Gamma$, we set a high value to enforce complete liquidation by the end of the trading day. For the risk aversion parameter $\gamma$, we choose an intermediate value that does not neutralize any of the financial effects our model could illustrate. Choosing a too high value of the risk aversion parameter would force the trader to liquidate quickly with no illustration of the impact of mean reversion. On the contrary, choosing a too low value of the risk aversion parameter would result in the trader accumulating unrealistically large positions to benefit from mean reversion at the expense of the original liquidation problem.\\

The resulting values used to run our algorithms are given in Table \ref{table:1DParams}.\footnote{In the one-asset case, $\Sigma = VV^\intercal$ is a scalar. We classically write it as $\sigma^2$ and document the value of $\sigma$.}\\ 

\begin{table}[H]
\begin{center}
\begin{tabular}{c  c} 
 \hline 
Parameter & Value \\ [0.5ex] 
 \hline
 $T$ & $1\ \textrm{day}$ \\ [0.5ex] 
 $q_0$ & $2250$ \\ [0.5ex]
 $S_0$ & $\$  79835 $ \\ [0.5ex] 
$R$ & $5.1\ \textrm{day}^{-1}$ \\ [0.5ex] 
$\overline S$ & $\$ 79887$ \\ [0.5ex] 
$\sigma$ & $243.67\ \$ \cdot \textrm{day}^{- \frac 12}$ \\ [0.5ex] 
 $\eta$ & $ 5 \cdot 10^{-3} \ \$ \cdot \textrm{day}$ \\ [0.5ex] 
  $\Gamma$ & $\$ 100\ $ \\ [0.5ex] 
$\gamma$ & $2 \cdot 10^{-5}\ \textrm{\$} ^{-1}$ \\
 \hline 
\end{tabular}
\end{center}
\caption {Value of the parameters.}
\label{table:1DParams}
\end{table}

We plot in Figure \ref{Inventory_1d_CAD} the asset price trajectory $(S_t)_{t \in [0,T]}$ on August 13, 2021 and the inventory process $(q_t)_{t \in [0,T]}$ corresponding to the use of the optimal liquidation strategy derived in the previous section.\footnote{In what follows, this strategy is often referred to as ACOU (Almgren-Chriss under Ornstein-Uhlenbeck dynamics) strategy.} For comparison purposes, we also plot the inventory process when using a classical Almgren-Chriss (AC) strategy.\footnote{To compute the AC strategy, we estimate the parameter $V_\textrm{AC}$ of the Bachelier dynamics. This parameter is a scalar in our one-asset case and we denote it by $\sigma_\textrm{AC}$ instead of $V_\textrm{AC}$. A simple estimation  based on price increments leads to $\sigma_\textrm{AC} = 244.02 \ \$  \cdot \textrm{day}^{- \frac 12}$ which slightly differs from $\sigma$ because the drift term in the OU model captures part of the variance.}\\

The results shown in Figure \ref{Inventory_1d_CAD} deserve several remarks. First, the optimal liquidation strategy in our model with mean reversion is different from that derived in the Almgren-Chriss model. In particular, the liquidation process is significantly faster in the latter case because the process of unwinding the portfolio appears far riskier to a trader who believes that the price evolves as a Brownian motion than to another one who believes in a mean-reverting Ornstein-Uhlenbeck dynamics. Second, in the case of the ACOU strategy, the trader progressively unwinds her long position over the trading day but also takes advantage of mean reversion. When the price is below $\overline S$, the trader tends to reduce the pace of her selling process or even buys some contracts. Symmetrically, when the price is above $\overline S$, the trader tends to sell at a faster pace. One exception to the above should nevertheless be noticed, close to time $T$. Indeed, because of the high value of the final penalty, close to time $T$ the trader focuses more on liquidating her portfolio and cares less about price oscillations. In particular, she buys back some contracts to end up flat because she previously went short to bet on the reversion of the price towards~$\overline S$.

\begin{figure}[!h]\centering
\includegraphics[width=0.97\textwidth]{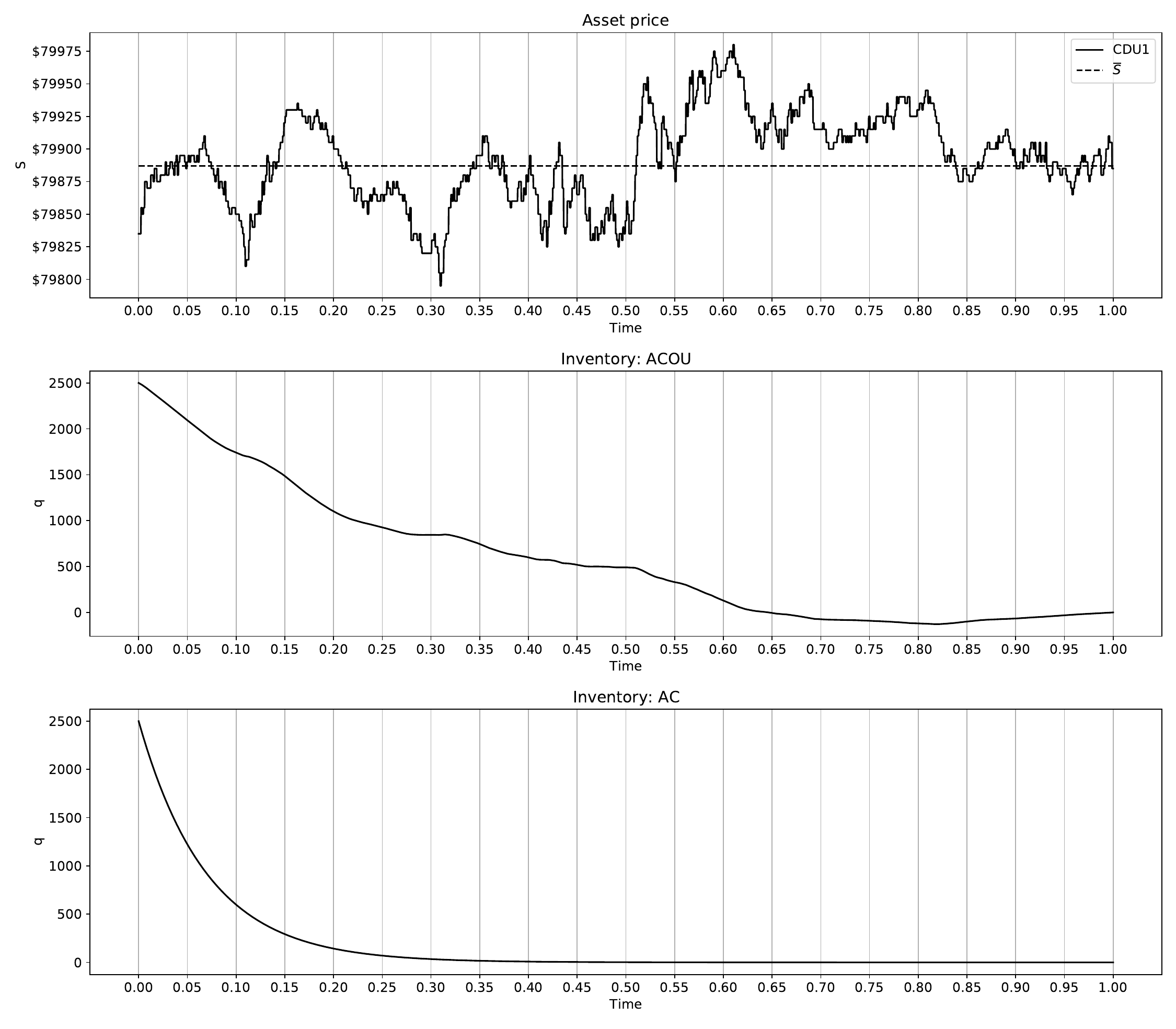}\\
\caption{Top: CDU1 price trajectory on August 13, 2021 -- $(S_t)_{t \in [0,T]}$. Middle: Trajectory of the inventory when using the optimal strategy corresponding to the estimated Ornstein-Uhlenbeck process -- $(q_t)_{t \in [0,T]}$. Bottom: Trajectory of the inventory when using the optimal strategy corresponding to a Brownian motion (Bachelier) model for the price (classical Almgren-Chriss strategy).}\label{Inventory_1d_CAD}
\end{figure}

\subsubsection{Impact of mean reversion}

We have just seen that mean reversion plays a key role in the characteristics of the strategy. In order to further study the impact of the mean-reversion parameter $R$, we consider the same parameters as in Table \ref{table:1DParams} except that we force the mean-reversion parameter $R$ to have the following values: $R = 0\ \textrm{day}^{-1}$, $R =3\ \textrm{day}^{-1}$, and $R = 10\ \textrm{day}^{-1}$. Moreover, unlike what we did previously, we consider now a price trajectory in line with the choice of the parameter $R$. For that purpose, we simulate price trajectories by using the same path of the Brownian motion but different values of $R$.\\

In Figure \ref{S_1d_exec}, we plot an instance of the (simulated) trajectories of the price process $(S_t)_{t\in [0,T]}$ for the different values of~$R$ and the corresponding inventory processes $(q_t)_{t \in [0,T]}$ -- along with the associated trading volume (or trading speed) processes $(v_t)_{t \in [0,T]}$ -- when using the optimal execution strategy.\\

\begin{figure}[!h]\centering
\includegraphics[width=\textwidth]{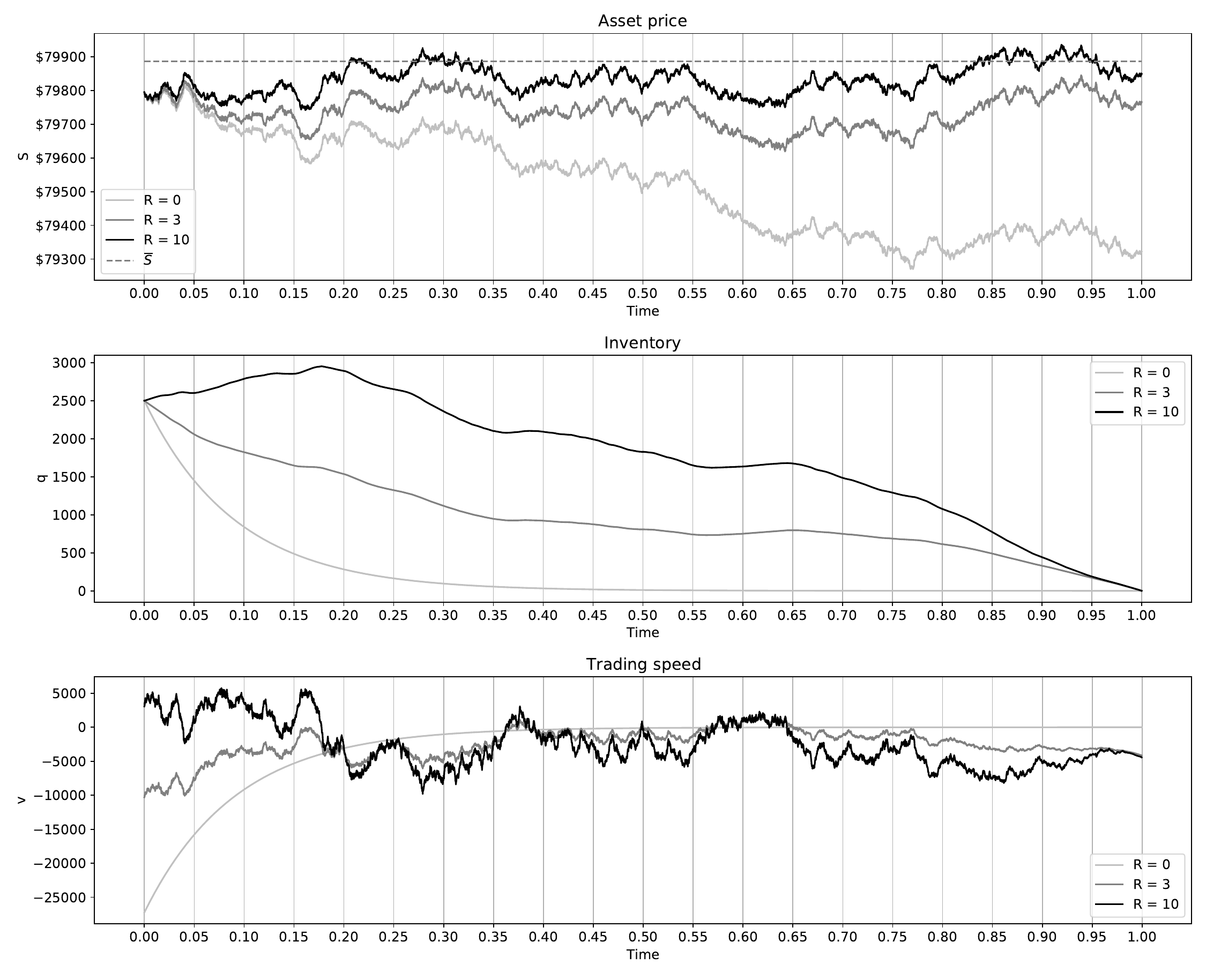}\\
\caption{Top: Simulated trajectories of the asset price for the different values of $R$ -- $(S_t)_{t \in [0,T]}$. Middle: Corresponding inventory processes when using the optimal execution strategy -- $(q_t)_{t \in [0,T]}$. Bottom: Trajectories of the corresponding trading volume (or trading speed) processes -- $(v_t)_{t \in [0,T]}$.}\label{S_1d_exec}
\end{figure}

The results of Figure \ref{S_1d_exec} confirm that the way the optimal strategy handles risk depends strongly on the mean-reversion parameter $R$. In particular, when $R$ is large, the trader acts almost\footnote{Close to time $T$, because of the high value of $\Gamma$, the trader focuses more on unwinding her portfolio and cares less about price oscillations, as can clearly be seen in Figure \ref{S_1d_exec}.} as if she was performing a VWAP/TWAP\footnote{VWAP and TWAP respectively mean Volume-Weighted Average Price and Time-Weighted Average Price. VWAP and TWAP strategies are commonly used by traders to execute orders at a price as close as possible to the average price of all transactions over a given period.} execution strategy plus a mean-reverting statistical arbitrage strategy. In particular, in the case where $R = 10\ \textrm{day}^{-1}$, the process $(v_t)_{t \in [0,T]}$ oscillates around its average (which is of course fixed by the total number of contracts to sell). These oscillations are highly correlated with those of $(S_t)_{t \in  [0,T]}$: the trader sells faster when the price is above $\overline S$ and slower (she even buys sometimes) when it is below~$\overline S$.

\subsubsection{Statistical arbitrage}

Given the observations of the previous subsection, it is natural to illustrate how our model can be used to build a statistical arbitrage strategy. For that purpose, we consider a trader with no initial inventory who starts trading the futures contract CDU1 on August 13, 2021 and wants to maximize the expected utility of her PnL at the end of the day (with no final penalty).\\

To run our algorithm we use the same parameters as in Table \ref{table:1DParams} except that $\Gamma = 0$ and $q_0 = 0$. The results are plotted in Figure \ref{Inventory_1d_StartArb_CAD}: the price process  $(S_t)_{t \in [0,T]}$, the inventory process $(q_t)_{t \in [0,T]}$ when using the optimal strategy, and the associated trajectory of the PnL, i.e. the process $(X_t + q_t S_t)_{t \in [0,T]}$.\\ 

\begin{figure}[!h]\centering
\includegraphics[width=\textwidth]{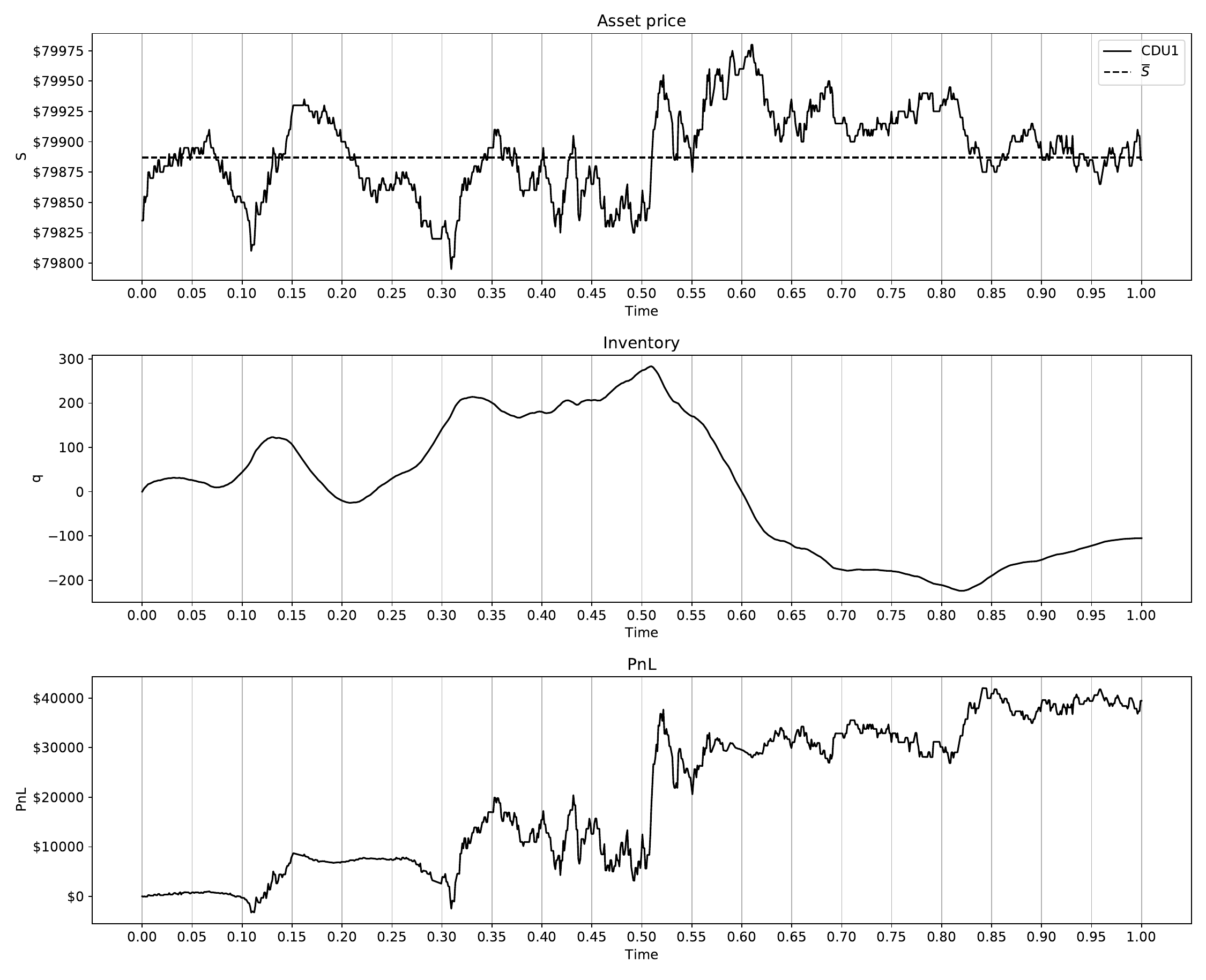}\\
\caption{Top: CDU1 price trajectory on August 13, 2021 -- $(S_t)_{t \in [0,T]}$. Middle: Trajectory of the inventory when using the optimal strategy starting from $q_0 = 0$ -- $(q_t)_{t \in [0,T]}$. Bottom: Corresponding trajectory of the PnL -- $(X_t + q_t S_t)_{t \in [0,T]}$.}\label{Inventory_1d_StartArb_CAD}
\end{figure}
\newpage

Because the trader believes that the price mean reverts around $\overline{S}$, her optimal strategy, in the absence of execution costs, would consist in having a long position when the price is below $\overline{S}$ and a short position when the price is above~$\overline{S}$. The optimal strategy when execution costs are taken into account is however more complex because the control of execution costs introduces inertia in the position of the trader. It consists instead in trading progressively to target a long position when the price is below $\overline{S}$ and a short position when the price is above $\overline{S}$. However, because of inertia, it happens that the position remains long while the price is still far above $\overline{S}$, as exemplified in Figure \ref{Inventory_1d_StartArb_CAD}. We see nevertheless that the strategy would have been profitable to the trader.\\

Another interesting experiment for assessing the performance of the strategy, when used for pure statistical arbitrage, consists in testing it on simulated price trajectories (using the same Ornstein-Uhlenbeck parameters as above for both simulating prices and computing the optimal strategies). We plot in Figure \ref{MtMval_StatArb_1d}, the distribution of the final PnL after $1500$ simulations of the price process. We see that our strategy allows to make money by taking advantage of the mean reversion: we get a positive final profit of $\$107698$ on average, with a standard deviation of $\$57791$, and the distribution looks skewed to the right (towards profits rather than losses).\\

\begin{figure}[!h]\centering
\includegraphics[width=\textwidth]{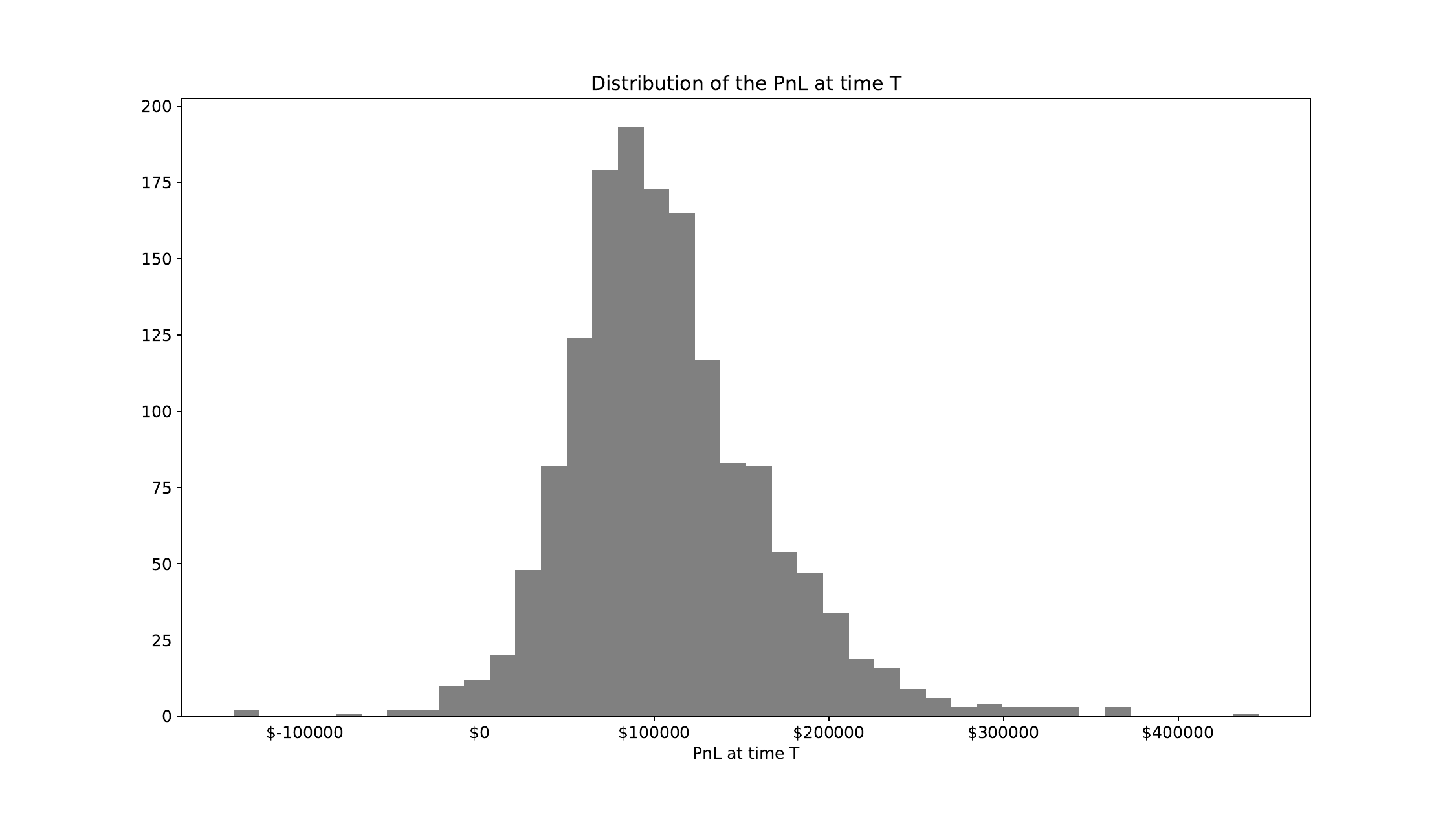}\\
\caption{Distribution of the final PnL for the statistical arbitrage strategy on CDU1 (for $1500$ simulations).}\label{MtMval_StatArb_1d}
\end{figure}

\subsection{Multi-asset case}

We now come to the use of our optimal strategies in the multi-asset case when asset prices exhibit a cointegrated behaviour that can be modeled by a multi-OU process. For that purpose, we use data from two French stocks within the banking sector: BNP Paribas (hereafter BNP) and Société Générale (hereafter GLE).\\

We plot in Figure \ref{S_BNP_SG} the mid-prices of BNP and GLE sampled every 60 seconds during the regular trading hours ($09$:$00$-$17$:$30$) over the week August 09-August 13, 2021. We clearly see that the stock prices of the two companies are driven by the same factors and should be cointegrated.\footnote{The existence of a cointegration vector is confirmed by a Johansen's cointegration test (see below).}\\

\begin{figure}[!h]\centering
\includegraphics[width=\textwidth]{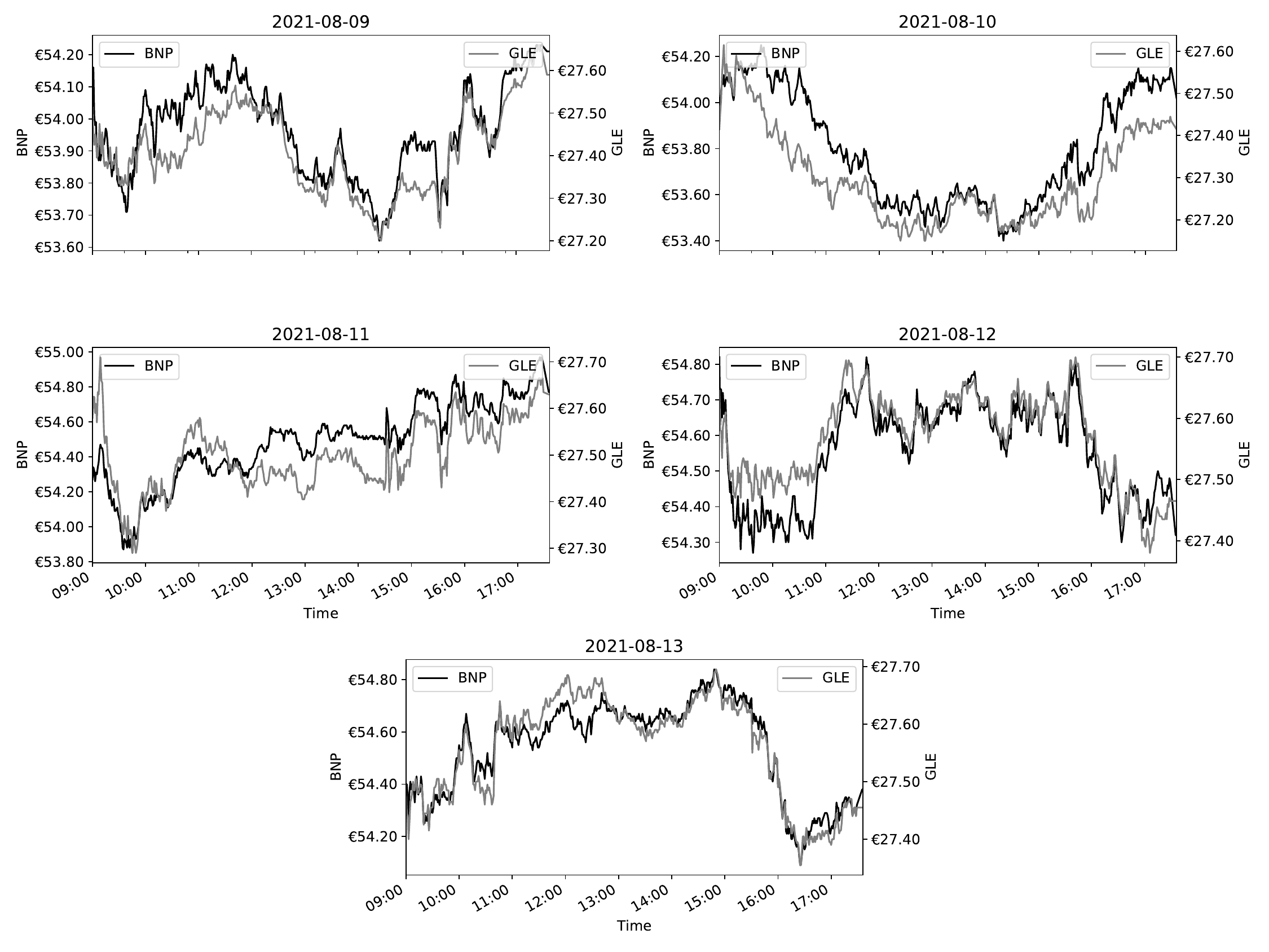}\\
\caption{Mid-prices of BNP (left axis) and GLE (right axis) sampled every 60 seconds during the regular trading hours ($09$:$00$-$17$:$30$) over the week August 09-August 13, 2021.}\label{S_BNP_SG}
\end{figure}
\vspace{5mm}
\subsubsection{Portfolio liquidation in the presence of cointegration}

We consider the case of a trader wishing to unwind a portfolio with $75000$ shares of BNP and $75000$ shares of GLE on August 13, 2021.\footnote{This represents roughly $5\%$ of the average daily traded volume over the period considered in this example.}\\

Similarly to our one-asset example, we consider that the trader estimates the parameters of a multi-OU model using prices from the four preceding trading days, here August 9, 10, 11, and 12, 2021. These parameters are estimated using classical linear regression techniques.\footnote{A time discretization of a multi-OU model gives rise to a Vector Auto-Regressive model of order $1$, or $\textrm{VAR}(1)$. The parameters of a $\textrm{VAR}(1)$ model are classically estimated by using least squares regression. Conversion of $\textrm{VAR}(1)$ coefficients into their continuous-time multi-OU counterparts is straightforward.} Referring to BNP and GLE by respectively using the superscripts $1$ and $2$, the estimated values of the parameters are given in Table \ref{table:MOUParams}.

\begin{table}[H]
\begin{center}
\begin{tabular}{c  c} 
 \hline 
Parameter & Estimate \\ [0.5ex] 
 \hline
 $R$ & $\begin{pmatrix} 0.33 & 3.95 \\ -2.52 & 10.23 \end{pmatrix} \ \textrm{day}^{-1}$ \\ [0.5ex] 
$\overline S$ & $\left(\overline S^1, \overline S^2 \right) = \left(\textrm{\euro} 54.23, \textrm{\euro} 27.45 \right) $ \\  [0.5ex] 
 $\Sigma$ & $\begin{pmatrix} 0.47 & 0.20 \\ 0.20 & 0.14 \end{pmatrix} \textrm{\euro}^2 \cdot \textrm{day}^{-1}$ \\
 \hline 
\end{tabular}
\end{center}
\caption {Multi-OU estimated parameters for the pair (BNP, GLE).}
\label{table:MOUParams}
\end{table}

The use of a Johansen's cointegration test\footnote{We use the Trace test and not the Maximum Eigenvalue test in what follows.} rejects the hypothesis of no cointegration but does not reject a cointegration rank $r = 1$ for the pair (BNP, GLE).\footnote{Johansen's approach is based on the VAR($1$) formulation $\Delta S_t = a + \Pi S_{t-1} + \epsilon_t$, where $a \in \mathbb R^d$, $\Pi \in \mathcal M_d (\mathbb R)$ and $\epsilon$ is normally distributed. In a nutshell, it iteratively tests the null hypothesis $\textrm{rank}(\Pi) \le k$ (corresponding to the existence of at most $k$ linearly independent cointegration vectors) for different values of $k$ using likelihood ratio statistics that follow tabulated distributions (see \cite{johansen1991estimation} for more details). In practice, the cointegration rank retained is the first value of $k$ for which the null hypothesis is not rejected.} Furthermore, the estimated value of the matrix $R$ suggests that the space of cointegration vectors is spanned by $(1, -3.46)$.\\

We give in Table \ref{table:johansenTestResultsTable} the detailed results for the Johansen's cointegration Trace test.\\

\begin{table}[H]
\begin{center}
\begin{tabular}{c c c c} 
 \hline
 Null Hypothesis &  Trace statistics & Critical Value & Conclusion\\ [0.5ex] 
 \hline
$r \leq 0$ & 16.77 & 15.49 & Rejected \\ [0.5ex]  
$r \leq 1$ & 2.293 & 3.841 & Not rejected \\ [1ex] 
 \hline
\end{tabular}
\end{center}
\caption {Johansen's cointegration Trace test results (critical values are given for a significance level of $95\%$).}
\label{table:johansenTestResultsTable}
\end{table}
\vspace{4mm}
To exemplify the use of our strategies and illustrate the different effects in the multi-asset case, we run our algorithms for two different values of the risk aversion parameter~$\gamma$. More precisely, we consider the parameters stated in Table \ref{table:MDParams} (values of the multi-OU parameters are not recalled, see Table \ref{table:MOUParams}).\footnote{The same logic as in the one-asset case has been applied for the choice of the parameters.}\\

\begin{table}[H]
\begin{center}
\begin{tabular}{c  c} 
 \hline 
Parameter & Value \\ [0.5ex] 
 \hline
 $T$ & $1\ \textrm{day}$ \\ [0.5ex]
 $q_0$ & $(75000, 75000)$ \\ [0.5ex]
 $S_0$ & $\left(S_0^1, S_0^2 \right) = \left(\textrm{\euro} 54.4, \textrm{\euro} 27.48 \right) $ \\ [0.5ex] 
 $\eta$ & $\begin{pmatrix} 4 \cdot 10^{-7}  & 0 \\ 0 & 2 \cdot 10^{-7} \end{pmatrix} \ \textrm{\euro} \cdot \textrm{day}$  \\ [0.5ex] 
  $\Gamma$ & $\textrm{\euro} 100 \times I_2$ \\ [0.5ex]
  $\gamma$ & $2 \cdot 10^{-5}\ \textrm{\euro} ^{-1}$ or $2 \cdot 10^{-3}\ \textrm{\euro} ^{-1}$ \\   
 \hline 
\end{tabular}
\end{center}
\caption {Value of the parameters.}
\label{table:MDParams}
\end{table}
\vspace{4mm}
We plot in Figures \ref{ACOU_OU_COMP_BNPGLE_SmallGamma} and \ref{ACOU_OU_COMP_BNPGLE_BigGamma} the trajectory on August 13, 2021 of the price process $(S_t)_{t \in [0,T]}$ and the spread process $((S^1_t - \overline{S}^1) - 3.46 (S^2_t - \overline{S}^2))_{t \in [0,T]}$ corresponding to the cointegration vector, along with the inventory process $(q_t)_{t \in [0,T]}$ corresponding to the use of the optimal liquidation strategy with $\gamma = 2 \cdot 10^{-5}\ \textrm{\euro} ^{-1}$ and $\gamma = 2 \cdot 10^{-3}\ \textrm{\euro} ^{-1}$ respectively.\footnote{As above, this strategy is referred to as ACOU (Almgren-Chriss under multivariate Ornstein-Uhlenbeck dynamics) strategy.} For comparison purposes, we also plot the inventory process when using a classical Almgren-Chriss (AC) strategy.\footnote{To compute the AC strategy, we estimate the parameter $V_\textrm{AC}$ of the Bachelier dynamics. A simple estimation based on price increments leads to $\Sigma_\textrm{AC} = V_\textrm{AC}{V_\textrm{AC}}^\intercal = \begin{pmatrix} 0.48 & 0.19 \\ 0.19 & 0.13 \end{pmatrix}\ \textrm{\euro}^2 \cdot \textrm{day}^{-1}$ which slightly differs from $\Sigma$ because the drift term in the OU model captures part of the variance.}
\\

\begin{figure}[!h]\centering
\includegraphics[width=0.95\textwidth]{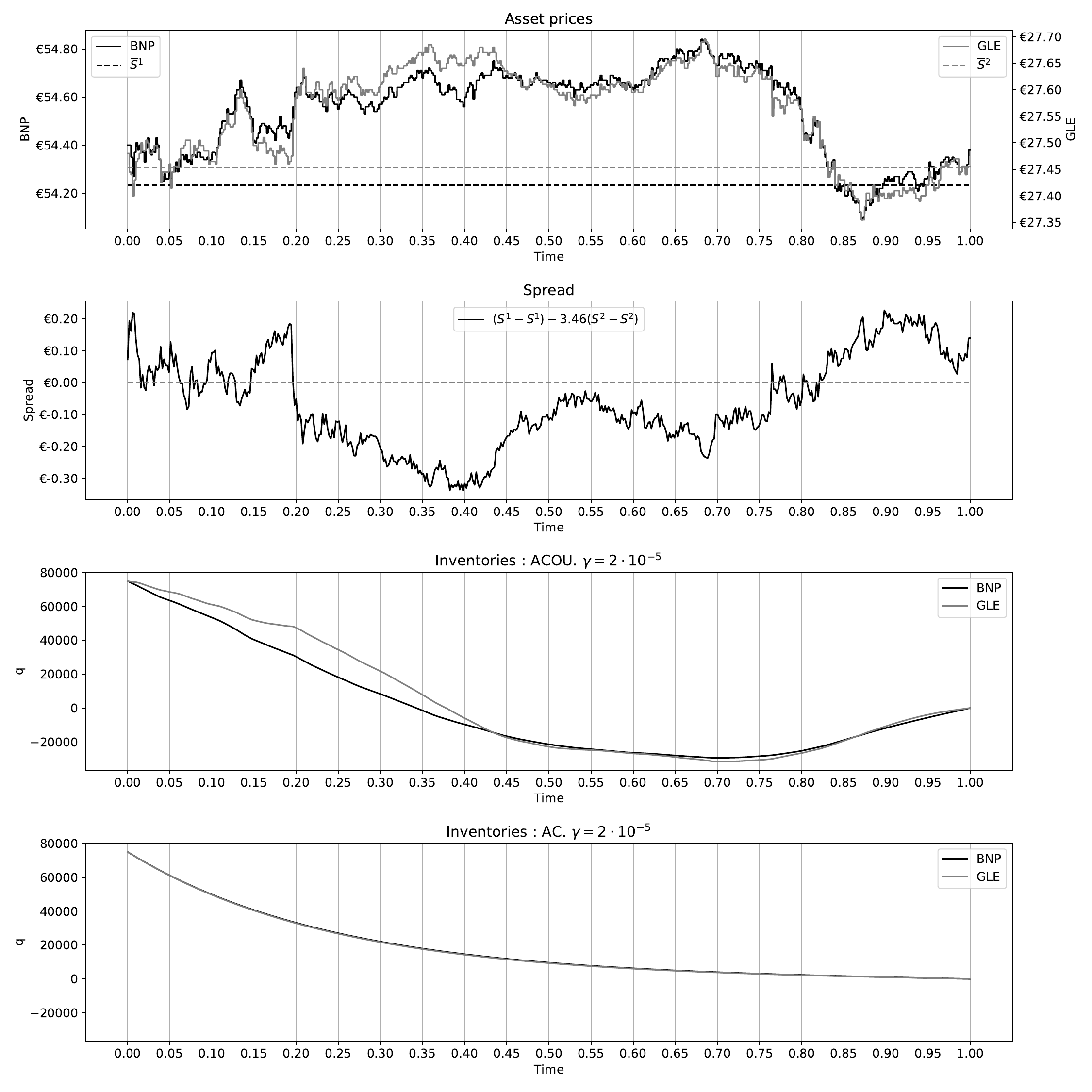}
\caption{Top two plots: BNP and GLE price trajectories on August 13, 2021 -- $(S_t)_{t \in [0,T]}$ -- and trajectory of the spread corresponding to the cointegration vector on August 13, 2021 -- $((S^1_t - \overline{S}^1) - 3.46 (S^2_t - \overline{S}^2))_{t \in [0,T]}$. Bottom two plots: Trajectory of the inventories when using the optimal strategy corresponding to the estimated multi-OU process with $\gamma = 2 \cdot 10^{-5}\ \textrm{\euro} ^{-1}$ -- $(q_t)_{t \in [0,T]}$ -- and when using the optimal strategy corresponding to correlated Brownian motions for the stock prices (classical Almgren-Chriss strategy with $\gamma = 2 \cdot 10^{-5}\ \textrm{\euro} ^{-1}$).}\label{ACOU_OU_COMP_BNPGLE_SmallGamma}
\end{figure}

\begin{figure}[!h]\centering
\includegraphics[width=0.92\textwidth]{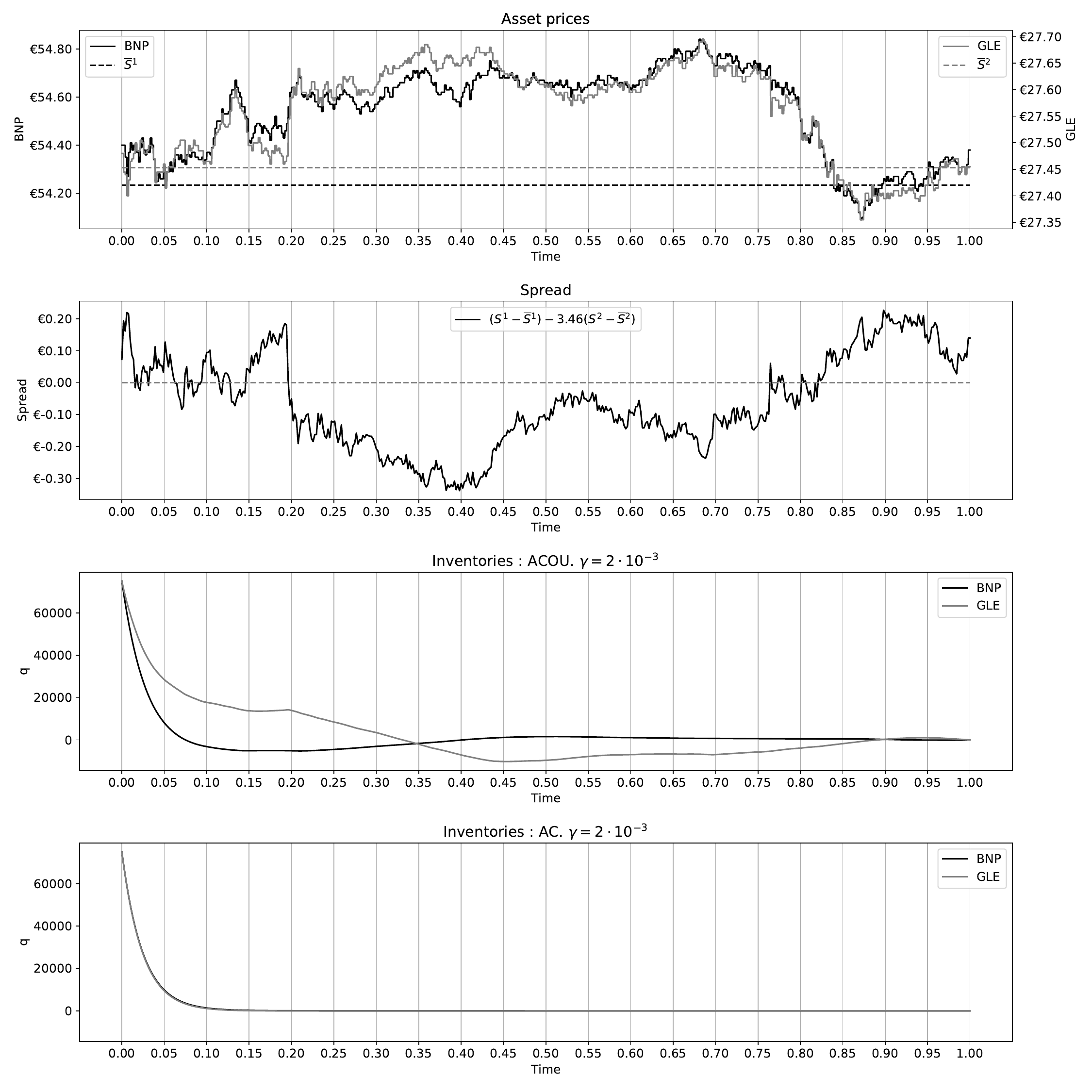}
\caption{Top two plots: BNP and GLE price trajectories on August 13, 2021 -- $(S_t)_{t \in [0,T]}$ --, and trajectory of the spread corresponding to the cointegration vector on August 13, 2021 -- $((S^1_t - \overline{S}^1) - 3.46 (S^2_t - \overline{S}^2))_{t \in [0,T]}$. Bottom two plots: Trajectory of the inventories when using the optimal strategy corresponding to the estimated multi-OU process with $\gamma = 2 \cdot 10^{-3}\ \textrm{\euro} ^{-1}$ -- $(q_t)_{t \in [0,T]}$ --, and when using the optimal strategy corresponding to correlated Brownian motions for the stock prices (classical Almgren-Chriss strategy with $\gamma = 2 \cdot 10^{-3}\ \textrm{\euro} ^{-1}$).}\label{ACOU_OU_COMP_BNPGLE_BigGamma}
\end{figure}

We clearly see, for both values of $\gamma$ that the ACOU strategy is different from the AC one\footnote{It is noteworthy that the trading curves for BNP and GLE are almost the same in the AC case. Because of the value of $S_0$ and $\eta$ this is unsurprising.} although they both succeed in unwinding the portfolio. This is easily understandable: because of the presence of the matrix $R$ in the multi-OU model, there is a drift in the dynamics of the prices that can be exploited to make money (while still controlling market risk).\\

It is interesting to understand the difference between what happens when $\gamma$ is small versus what happens when $\gamma$ is large because the trading curves exhibit very different properties. When $\gamma$ is small, we observe that the trader oversells the two stocks (before buying back close to time $T$ to unwind the portfolio) and speculates therefore on the reversion of the two stock prices towards $\overline{S}^1$ and $\overline{S}^2$ respectively -- because prices are above these values. On the contrary, when $\gamma$ is large, shorting the two stocks simultaneously appears too risky and, once the portfolio has been partially liquidated, the trader uses instead a long/short strategy. In that case, the trader seems in fact to apply a statistical arbitrage strategy related to the spread process $((S^1_t - \overline{S}^1) - 3.46 (S^2_t - \overline{S}^2))_{t \in [0,T]}$: for $t \ge 0.35$, when the spread process is below~$0$, the trader is long BNP and short GLE -- she buys the spread -- while it becomes (slightly) long GLE and (slightly) short BNP -- she shorts the spread -- at the very end of the period when the spread becomes positive. The fact that the trader ``trades the spread'' will appear even more clearly in what follows as we focus on statistical arbitrage strategies.\\ 
\vspace{-4mm}
\subsubsection{Statistical arbitrage}

Given the previous remarks, it is natural to illustrate the use of our model in the context of pure statistical arbitrage. For that purpose, we consider a trader with no initial inventory who starts trading on August 13, 2021 and wants to maximize the expected utility of her PnL at the end of the day (with no final penalty).\\

\begin{figure}[!h]\centering
\includegraphics[width=0.85\textwidth]{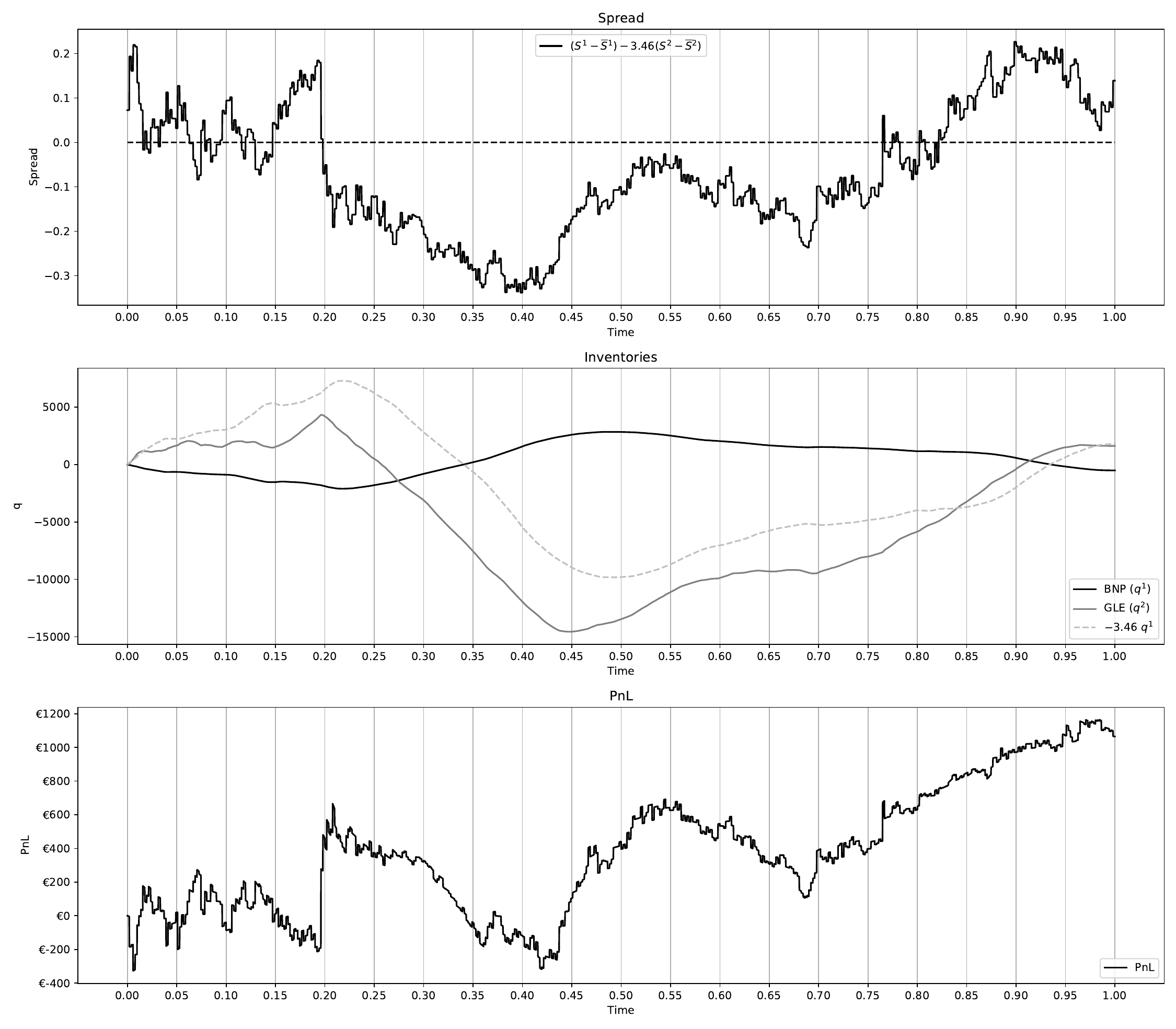}\\
\caption{Top: Trajectory of the spread on August 13, 2021 -- $((S^1_t - \overline{S}^1) - 3.46 (S^2_t - \overline{S}^2))_{t \in [0,T]}$. Middle: Trajectory of the inventories when using the optimal strategy corresponding to the estimated multi-OU process with $\gamma = 2 \cdot 10^{-3}\ \textrm{\euro} ^{-1}$ -- $(q_t)_{t \in [0,T]}$. Bottom: Trajectory of the PnL -- $(X_t + q^\intercal_t S_t)_{t \in [0,T]}$.} \label{ACOU_OU_COMP_BNPGLE_SYNTH}
\end{figure}

To run our algorithm we use the same parameters as in Tables \ref{table:MOUParams} and \ref{table:MDParams} with $\gamma = 2 \cdot 10^{-3}\ \textrm{\euro} ^{-1}$ except that $\Gamma = 0$ and $q_0 = (0,0)$. The results are plotted in Figure
 \ref{ACOU_OU_COMP_BNPGLE_SYNTH}: the spread process $((S^1_t - \overline{S}^1) - 3.46 (S^2_t - \overline{S}^2))_{t \in [0,T]}$, the inventory process $(q_t)_{t \in [0,T]}$ when using the optimal strategy, and the associated trajectory of the PnL, i.e. the process $(X_t + q^\intercal_t S_t)_{t \in [0,T]}$.\\
 
We clearly see that for $\gamma = 2 \cdot 10^{-3}\ \textrm{\euro} ^{-1}$ the optimal strategy is a long/short strategy. As can be seen in Figure
 \ref{ACOU_OU_COMP_BNPGLE_SYNTH} (middle plot), the process $(-3.46 q^1_t)_{t \in [0,T]}$ appears to be in line with $(q^2_t)_{t \in [0,T]}$. This confirms that the strategy consists mainly in ``buying or selling the spread'' depending on the sign of the spread process $((S^1_t - \overline{S}^1) - 3.46 (S^2_t - \overline{S}^2))_{t \in [0,T]}$.\\ 

Finally, as in the previous subsection, we test our optimal strategy on simulated price trajectories (using the same multi-OU parameters as above for both simulating prices and computing the optimal strategies). We plot in Figure~\ref{MtMval_StatArb_md}, the distribution of the final PnL after $1500$ simulations of the price process when using the optimal strategy with the parameters of Table \ref{table:MDParams} (with $\gamma = 2 \cdot 10^{-3}\ \textrm{\euro} ^{-1}$) except that $\Gamma = 0$ and $q_0 = (0,0)$ (because we focus on statistical arbitrage). We see that our strategy allows to make money by taking advantage of the price dynamics: we get a positive final profit of $\textrm{\euro} 2230$ on average, with a standard deviation of $\textrm{\euro} 1145$, and the distribution is, as above, skewed towards profits rather than losses.\footnote{The PnLs are smaller in absolute value than in the one-asset example because the value of $\gamma$ is higher here.}\\

\begin{figure}[!h]\centering
\includegraphics[width=0.97\textwidth]{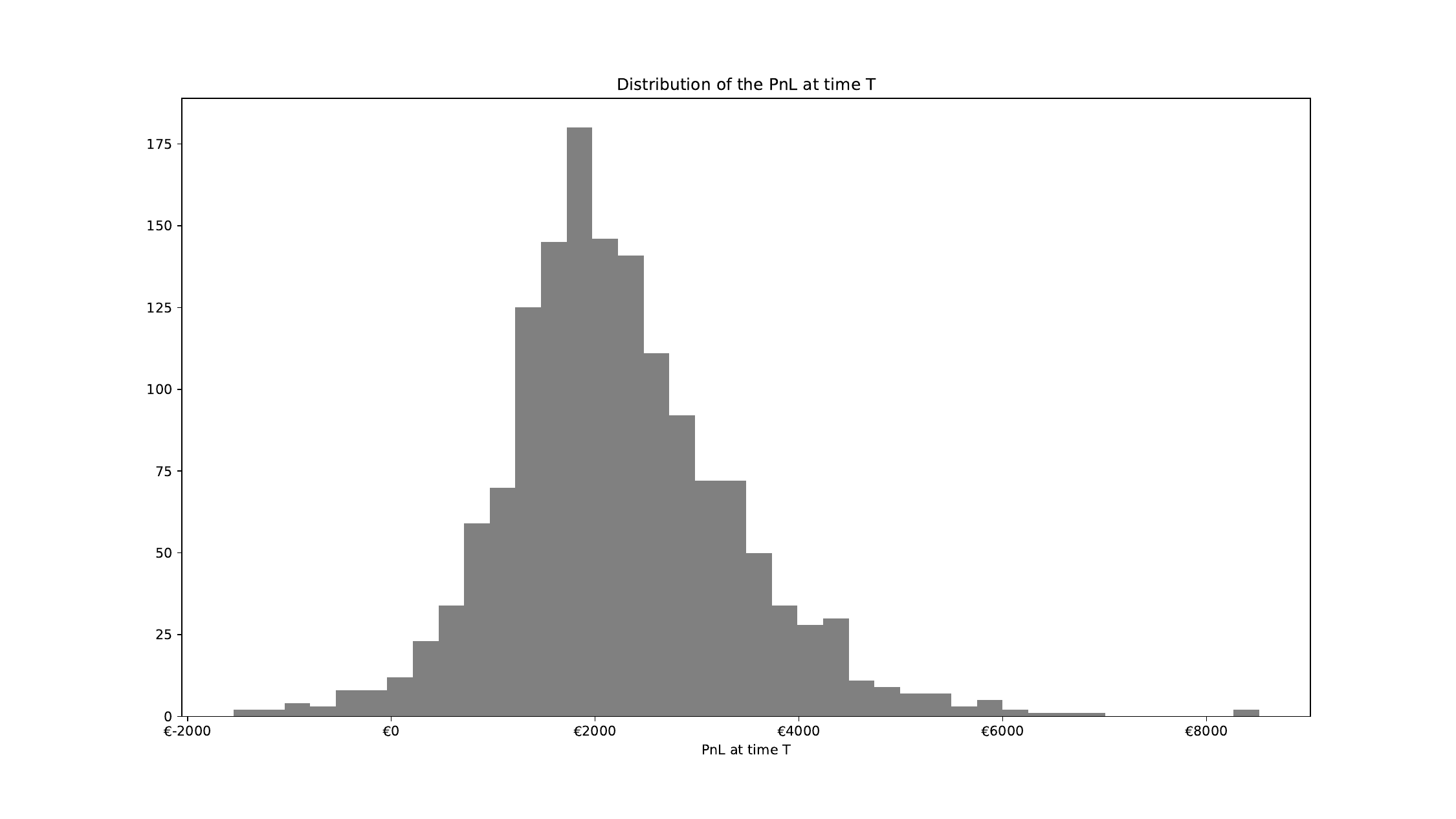}\\
\caption{Distribution of the final PnL for the statistical arbitrage strategy on BNP and GLE (for $1500$ simulations).}\label{MtMval_StatArb_md}
\end{figure}
\vspace{-2mm}
\section*{Conclusion}

In this paper, we have shown how to account for cross-asset co-movements when executing trades in multiple assets. In our model, the agent has an exponential utility and the prices have multivariate Ornstein-Uhlenbeck dynamics, capturing the complex cross-asset dynamics of prices better than correlated Brownian motions only. The advantage of our approach is twofold: (i) it better accounts for risk at the portfolio level, and (ii) it is versatile and can be used for basket execution and statistical arbitrage.\\

The advantages for practitioners are numerous. Considering asset execution within a portfolio allows to manage risk across a wider basket of assets rather than considering only the risk of a single trade. Agents can hold securities on their balance sheets for longer, reducing market impact and execution costs. Moreover, from a regulation point of view, multivariate optimal execution models that naturally offset risks in a portfolio are of great interest. In fact, the new FRTB (Fundamental Review of the Trading Book) regulation will lead practitioners to assess liquidity risks within a centralized risk book for capital requirements. In this context, our model can reduce the liquidity risk of the execution process by taking into account the joint dynamics of the assets.\\

\clearpage
\appendix

\section{Appendix -- Multi-asset optimal execution with correlated Brownian motions and execution costs} \label{annex1}

We consider in this appendix the problem of multi-asset optimal execution in the case where prices are correlated arithmetic Brownian motions. This problem is a special case of that presented in this paper, corresponding to $R = 0$ in the dynamics \eqref{PriceProces} of the asset prices. Therefore, the results presented in the paper apply. However, when $R=0$, as mentioned in Remark \ref{remark1}, the system of ODEs \eqref{ODEsystem} simplifies since a trivial solution to the last five equations is $B=C=D=E=F=0$. Therefore, the problem boils down to finding $A \in C^1 \left([0,T], \mathcal S_d(\mathbb R) \right)$ solution of the following terminal value problem:

\begin{align}
\label{ODEsAnnexA}
\begin{cases}
A'(t) & =\frac{\gamma}{2}\Sigma-A(t)\eta^{-1}A(t) \\
A(T) & =-\Gamma. 
\end{cases}
\end{align}

In this appendix we show that, when $\Sigma \in S^{++}_d(\mathbb R)$, $A$ can be found in closed form.\\

For that purpose, we introduce the change of variables 
$$a(t)=\eta^{-\frac{1}{2}}A(t)\eta^{-\frac{1}{2}} \quad \forall t \in [0,T]$$ and notice that \eqref{ODEsAnnexA} is equivalent to the terminal value problem
\begin{align}
\begin{cases}
\label{ODEsPetitAAnnexA}
a'(t)= & \hat{A}^2-a(t)^{2}  \\
a(T)= & - C, 
\end{cases}
\end{align}

where $\hat{A}=\sqrt{\frac{\gamma}{2}} \left(\eta^{-\frac{1}{2}}\Sigma\eta^{-\frac{1}{2}}\right)^{\frac 12} \in \mathcal S^{++}_d(\mathbb R)$ and $C = \eta^{-\frac{1}{2}}\Gamma\eta^{-\frac{1}{2}} \in \mathcal S^{+}_d(\mathbb R)$. \\

To solve \eqref{ODEsPetitAAnnexA} we use a classical trick for Riccati equations in the following Proposition:
\begin{prop}
\label{propAnnexA}
Let $\xi : [0, T] \rightarrow \mathcal S_d( \mathbb R)$ defined as
\begin{align}
\label{defXiAnnexA}
    \xi\left(t\right) = -\frac{\hat A^{-1}}{2} \left(I - e^{-2 \hat{A} \left(T-t\right)} \right) -  e^{-\hat{A} \left(T-t\right)} \left(C + \hat A \right)^{-1} e^{- \hat{A} \left(T-t\right)}
\end{align}

be the unique solution of the linear ODE 

\begin{align}
\label{OdeXIAnnexA}
    \begin{cases}
    \xi'(t) = \hat A \xi(t) + \xi(t) \hat A + I_d \\
    \xi(T) = -\left(C + \hat A\right) ^{-1}.
    \end{cases}
\end{align}

Then $\forall t \in [0, T]$, $\xi(t)$ is invertible and $a: t \in [0,T] \rightarrow \hat A + \xi(t)^{-1} \in \mathcal S_d(\mathbb R)$  is the unique solution of \eqref{ODEsPetitAAnnexA}.\\

\end{prop}

\begin{proof}

First, we easily verify that $\xi$, defined in \eqref{defXiAnnexA}, is solution of the linear ODE \eqref{OdeXIAnnexA}. We see that, for all $t \in [0,T]$, $-\xi(t)$ is the sum of $\frac{\hat A^{-1}}{2} \left(I - e^{-2 \hat{A} \left(T-t\right)} \right) \in \mathcal S^{+}_d(\mathbb R)$ and $e^{-\hat{A} \left(T-t\right)} \left(C + \hat A \right)^{-1} e^{- \hat{A} \left(T-t\right)} \in \mathcal S^{++}_d(\mathbb R)$, so $-\xi(t) \in \mathcal S^{++}_d(\mathbb R)$ and $\xi(t)$ is invertible. \\

We also note that

$$a'(t)  = -\xi(t)^{-1} \xi'(t) \xi(t)^{-1} = -\xi(t)^{-1} \hat A - \hat A \xi(t)^{-1} - \xi(t)^{-2} = \hat A^2 - \left(\hat A + \xi(t)^{-1} \right)^2 = \hat A^2 - a(t)^2$$
and $a(T) = -C$, hence the result.
\end{proof}

We deduce the following corollary:
\begin{crl}
$$\forall t \in [0,T], \quad A(t) = \eta^{\frac12} \left(\hat A - \left( \frac{\hat A^{-1}}{2}  \left( I - e^{-2\hat A(T-t)} \right)    + e^{-\hat A (T-t)} \left(C + \hat A \right)^{-1} e^{-\hat A (T-t)}\right)^{-1} \right) \eta^{\frac12}. $$
\end{crl}

\section{Appendix -- Merton portfolio optimization problem under Ornstein-Uhlenbeck dynamics and exponential utility} \label{annexMerton}

\subsection{Modelling framework}

We study in this appendix a Merton problem where prices have multivariate Ornstein-Uhlenbeck dynamics. It is closely related to our problem and can be seen as some form of limit case corresponding to no execution costs (i.e. $L=0$) and no terminal penalty (i.e. $\ell = 0$).\\

The results obtained in this appendix are essential in our proof of existence of a solution to the system of ODEs \eqref{ODEsystem} on $[0,T]$ with terminal condition \eqref{termcondABCDEF} (see Theorem \ref{existence}).\\

As in the body of the paper,\footnote{We consider no permanent market impact in this Appendix.} we consider a model with $d$ assets, whose prices are modelled by a $d$-dimensional stochastic process $(S_t)_{t \in [0,T]} = \left(S^1_t, \ldots, S^d_t \right)_{t \in [0,T]} ^\intercal$ with dynamics
$$dS_t = R(\overline S - S_t) dt + VdW_t,$$
where $\overline S \in \mathbb R^d$, $R \in \mathcal M_d(\mathbb R)$, $V \in \mathcal M_{d,k}(\mathbb R)$, and $(W_t)_{t \in [0,T]} = \left(W^1_t, \ldots, W^k_t \right)^\intercal_{t \in [0,T]}$ is a $k$-dimensional standard Brownian motion (with independent coordinates), for some $k \in \mathbb{N}^*$. As before, we write $\Sigma = VV^\intercal$.\\

We consider a trader optimizing her portfolio over the period $[0,T]$ by controlling at each time the number of each asset in her portfolio, i.e. she controls the $d$-dimensional process $(q_t)_{t \in [0,T]} = \left(q^1_t, \ldots, q^d_t \right)_{t \in [0,T]}^\intercal$, where $q^i_t$ denotes the number of assets $i$ in the portfolio at time $t$, for each $i \in \{1,\ldots, d\}$ ($t \in [0,T]$).\footnote{Unlike what happens in the model of Section 2, our control variable is here the number of assets and not the volume traded. It is only when execution costs are incurred by the trader that trading rates / trading volumes are indeed the relevant control variables.} The process $(q_t)_{t \in [0,T]}$ lies in the space of admissible controls $\mathcal A^{Merton}_0$, where for $t \in [0,T]$, the set $\mathcal A^{Merton}_t$ is defined as 
\begin{align}
    \label{Amerton}
    \mathcal A^{Merton}_t := \left\{ (q_s)_{s \in [t,T]},\ \mathbb R^d\textrm{-valued},\ \mathbb F\textrm{-adapted, satisfying a linear growth condition with respect to } (S_s)_{s \in [t,T]}  \right\}.
\end{align}
We introduce the process $(\mathcal V_t)_{t \in [0,T]}$ modelling the MtM value of the trader's portfolio, i.e.
$$\forall t \in [0,T],\quad \mathcal V_t = \mathcal V_0 + \int_0^t q_s^\intercal dS_s, \qquad \mathcal V_0 \in \mathbb R \ \text{given}.$$

For a given $\gamma>0$, the trader aims at maximizing the following objective function:
\begin{align}
\label{objMerton}
\mathbb E \left[ - e^{-\gamma \mathcal V_T} \right],
\end{align}
over the set of admissible controls $(q_t)_{t \in [0,T]} \in \mathcal A^{Merton}_0.$ We define her value function $\hat u: [0,T] \times \mathbb R \times \mathbb R^d \rightarrow \mathbb R$ as
\begin{align*}
    \hat u(t,\mathcal V,S) = \underset{(q_s)_{s \in [t,T]} \in \mathcal A^{Merton}_t}{\sup} \mathbb E \left[ - e^{-\gamma \mathcal V^{t,\mathcal V,S,q}_T}\right] \qquad \forall (t,\mathcal V,S) \in [0,T] \times \mathbb R \times \mathbb R^d,
\end{align*}
where $(\mathcal V_s^{t,\mathcal V,S,q})_{s \in [t,T]}$ denotes the process defined by
$$d\mathcal V^{t,\mathcal V,S,q}_s = q_s^\intercal dS^{t,S}_s, \quad \mathcal V^{t,\mathcal V,S,q}_t = \mathcal V$$
with $$dS^{t,S}_s = R(\overline S - S^{t,S}_s) ds + VdW_s, \quad S^{t,S}_t = S.$$

\subsection{HJB equation}

The HJB equation associated with Problem \eqref{objMerton} is given by
\begin{align} 
\label{HJBMerton}
0 =  & \, \partial_t \hat w(t,\mathcal V,S) + \nabla_S \hat w(t,\mathcal V,S)^\intercal  R(\overline S - S) + \frac 12 \text{Tr} \left(\Sigma D^2_{SS} \hat w(t,\mathcal V,S) \right)\\
& +\underset{q \in \mathbb R^d}{\sup} \left\{ \partial_{\mathcal V} \hat w(t,{\mathcal V},S) q^\intercal R(\overline S-S) + \frac 12 \partial^2_{{\mathcal V}{\mathcal V}} \hat w(t,{\mathcal V},S) q^\intercal \Sigma q + \partial_{\mathcal V} \nabla_S \hat w(t,{\mathcal V},S)^\intercal \Sigma q \right\} \nonumber
\end{align}
for all $(t,{\mathcal V},S) \in [0,T) \times \mathbb R \times \mathbb R^d$, with terminal condition 
\begin{align}
\label{termcondMerton}
\hat w(T,{\mathcal V},S) = -e^{-\gamma {\mathcal V}} \quad \forall ({\mathcal V},S) \in \mathbb R \times \mathbb R^d.
\end{align}
To solve the above HJB equation, we use the ansatz
\begin{align}
\label{ansatz1Merton}
\hat w(t,{\mathcal V},S) = -e^{-\gamma \left({\mathcal V} + \hat \theta(t,S)\right)}.
\end{align}
Indeed, we have the following proposition:
\begin{prop}
If there exists $\hat \theta \in C^{1,2}([0 ,T] \times \mathbb R^d, \mathbb R)$ solution to
\begin{align}
\label{thetaMerton}
 0 =\ &\partial_t \hat \theta(t,S)  + \frac 12 \text{Tr} \left( \Sigma D^2_{SS} \hat \theta(t,S)  \right)  + \frac{1}{2\gamma } (\overline S - S)^\intercal R^\intercal \Sigma^{-1} R(\overline S-S)
\end{align}
on $[0 ,T)  \times \mathbb R^d$, with terminal condition
\begin{align}
\label{termcondthetaMerton}
\hat \theta(T,S) = 0 \qquad \forall S \in \mathbb R^d,
\end{align}
then the function $\hat w:[0,T] \times \mathbb R \times \mathbb R^d \rightarrow \mathbb R$ defined by 
$$\hat w(t,{\mathcal V},S) = -e^{-\gamma \left({\mathcal V} + \hat \theta(t,S)\right)} \quad \forall (t,{\mathcal V},S) \in [0,T] \times \mathbb R \times \mathbb R^d$$
is a solution to \eqref{HJBMerton} on $[0,T)\times \mathbb R \times \mathbb R^d $ with terminal condition \eqref{termcondMerton}.
\end{prop}
\begin{proof}
Let $\hat \theta:[0 ,T] \times \mathbb R^d  \rightarrow \mathbb R$ be a solution to \eqref{thetaMerton} on $[0,T) \times \mathbb R^d$ with terminal condition \eqref{termcondthetaMerton}, then we have for all $(t,{\mathcal V},S) \in [0,T) \times \mathbb R \times \mathbb R^d $:
\begin{align*}
& \partial_t \hat w(t,{\mathcal V},S) + \nabla_S \hat w(t,{\mathcal V},S)^\intercal  R(\overline S - S) + \frac 12 \text{Tr} \left(\Sigma D^2_{SS} \hat w(t,{\mathcal V},S)  \right)\nonumber\\
& +\underset{q \in \mathbb R^d}{\sup} \left\{ \partial_{\mathcal V} \hat w(t,{\mathcal V},S) q^\intercal R(\overline S-S) + \frac 12 \partial^2_{{\mathcal V}{\mathcal V}} \hat w(t,{\mathcal V},S) q^\intercal \Sigma q + \partial_{\mathcal V} \nabla_S \hat w(t,{\mathcal V},S)^\intercal \Sigma q \right\}  \\  
=\ & -\gamma \partial_t \hat \theta(t,S)\hat w(t,{\mathcal V},S) -\gamma\nabla_S \hat \theta (t,S)R(\overline S - S) \hat w(t,{\mathcal V},S) - \gamma \hat w(t,{\mathcal V},S) \frac 12 \text{Tr} \left(\Sigma D^2_{SS} \hat \theta(t,S)  \right)  \nonumber\\
& + \frac{\gamma^2}{2} \hat w(t,{\mathcal V},S) \nabla_S \hat \theta (t,S)^\intercal \Sigma \nabla_S \hat \theta (t,S)\\
& +\underset{q \in \mathbb R^d}{\sup} \left\{ -\gamma \hat w(t,{\mathcal V},S) q^\intercal R(\overline S-S) + \frac {\gamma^2}2 \hat w(t,{\mathcal V},S) q^\intercal \Sigma q + \gamma^2 \hat w(t,{\mathcal V},S) \nabla_S \hat \theta(t,{\mathcal V},S)^\intercal \Sigma q \right\}   \\
=\ & -\gamma \hat w(t,{\mathcal V},S) \Bigg( \partial_t \hat \theta(t,S)+\nabla_S \hat \theta (t,S)R(\overline S - S)+  \frac 12 \text{Tr} \left(\Sigma D^2_{SS} \hat \theta(t,S) \right) - \frac{\gamma}{2} \nabla_S \hat \theta (t,S)^\intercal \Sigma \nabla_S \hat \theta (t,S)\\
& +\underset{q \in \mathbb R^d}{\sup} \left\{  q^\intercal \left( R(\overline S-S) - \gamma \Sigma \nabla_S \hat \theta(t,S) \right) - \frac {\gamma}2  q^\intercal \Sigma q \right\} \Bigg).
\end{align*}
The supremum in the above equation is reached for $q = q^*(t,S) = \frac 1{\gamma} \Sigma^{-1}R(\overline S-S) - \nabla_S \hat \theta(t,S)$,
and we obtain therefore after simplifications:
\begin{align*}
& \partial_t \hat w(t,{\mathcal V},S) + \nabla_S \hat w(t,{\mathcal V},S)^\intercal  R(\overline S - S) + \frac 12 \text{Tr} \left(\Sigma D^2_{SS} \hat w(t,{\mathcal V},S)\right)\nonumber\\
& +\underset{q \in \mathbb R^d}{\sup} \left\{ \partial_{\mathcal V} \hat w(t,{\mathcal V},S) q^\intercal R(\overline S-S) + \frac 12 \partial^2_{{\mathcal V}{\mathcal V}} \hat w(t,{\mathcal V},S) q^\intercal \Sigma q + \partial_{\mathcal V} \nabla_S \hat w(t,{\mathcal V},S)^\intercal \Sigma q \right\}  \\  
=\ & -\gamma \hat w(t,{\mathcal V},S) \Bigg( \partial_t \hat \theta(t,S)+  \frac 12 \text{Tr} \left(\Sigma D^2_{SS} \hat \theta(t,S) \right) + \frac{1}{2\gamma } (\overline S - S)^\intercal R^\intercal \Sigma^{-1} R(\overline S-S) \Bigg) \\
=\ & 0.
\end{align*}
As $\hat w$ satisfies the terminal condition \eqref{termcondMerton}, the result is proved.
\end{proof}

We now use a second ansatz and look for a function $\hat \theta$ solution to \eqref{thetaMerton} on $[0,T) \times \mathbb R^d$ with terminal condition \eqref{termcondthetaMerton} of the following form:
\begin{align}
\label{ansatz2Merton}
\hat \theta(t,S) = S^\intercal \hat C(t) S + \hat E(t)^\intercal S + \hat F(t).
\end{align}
We have indeed the following proposition:
\begin{prop}
Assume there exists $\hat C \in C^1 \left([0,T], \mathcal S_d(\mathbb R) \right)$, $\hat E \in C^1 \left([0,T], \mathbb R^d \right)$, $\hat F \in C^1 \left([0,T], \mathbb R \right)$ satisfying the system of ODEs 
\begin{align}
\label{ODEMerton}
\begin{cases}
\hat C'(t) &= - \frac 1{2\gamma} R^\intercal \Sigma^{-1} R \\
\hat E'(t) &= \frac 1{\gamma} R^\intercal \Sigma^{-1} R \overline S \\
\hat F'(t) &= - \text{Tr} \left( \hat C(t) \Sigma \right) - \frac{1}{2\gamma} \overline S^\intercal R^\intercal \Sigma^{-1} R \overline S,
\end{cases}
\end{align}
with terminal condition
\begin{align}
\label{termcondCEFMerton}
\hat C(T) = \hat E(T) = \hat F(T) = 0.
\end{align}
Then the function $\hat \theta$ defined by \eqref{ansatz2Merton} satisfies \eqref{thetaMerton} on $[0,T) \times \mathbb R^d $ with terminal condition \eqref{termcondthetaMerton}.
\end{prop}
\begin{proof}
Let us consider $\hat C \in C^1 \left([0,T], \mathcal S_d(\mathbb R) \right)$, $\hat E \in C^1 \left([0,T], \mathbb R^d \right)$, $\hat F \in C^1 \left([0,T], \mathbb R \right)$ verifying \eqref{ODEMerton} on $[0,T)$ with terminal condition \eqref{termcondCEFMerton}. Let us consider $\hat \theta: [0,T] \times \mathbb R^d \rightarrow \mathbb R$ defined by \eqref{ansatz2Merton}. Then we obtain for all $(t,S) \in [0,T) \times \mathbb R^d$:
\begin{align*}
& \partial_t \hat \theta(t,S)  + \frac 12 \text{Tr} \left(\Sigma D^2_{SS} \hat \theta(t,S) \right)  + \frac{1}{2\gamma } (\overline S - S)^\intercal R^\intercal \Sigma^{-1} R(\overline S-S)\\
= \quad & S^\intercal \hat C'(t) S + {\hat E'(t)}^\intercal S + \hat F'(t) + \text{Tr} \left(\hat C(t) \Sigma \right)+ \frac{1}{2\gamma } (\overline S - S)^\intercal R^\intercal \Sigma^{-1} R(\overline S-S)\\
= \quad & 0.
\end{align*}
As it is straightforward to verify that $\hat \theta$ satisfies the terminal condition \eqref{termcondthetaMerton}, the result is proved.
\end{proof}

It is straightforward to see that there exists a unique solution $\hat C \in C^1 \left([0,T], \mathcal S_d(\mathbb R) \right)$, $\hat E \in C^1 \left([0,T], \mathbb R^d \right)$, $\hat F \in C^1 \left([0,T], \mathbb R \right)$ to \eqref{ODEMerton} with terminal condition \eqref{termcondCEFMerton}. We can then prove the following verification theorem.

\begin{thm}
We consider the functions $\hat C \in C^1 \left([0,T], \mathcal S_d(\mathbb R) \right)$, $\hat E \in C^1 \left([0,T], \mathbb R^d \right)$, $\hat F \in C^1 \left([0,T], \mathbb R \right)$ solutions to \eqref{ODEMerton} with terminal condition 
$$\hat C(T) = \hat E(T) = \hat F(T) = 0,$$
i.e. for all $t \in [0,T]$,
$$
\begin{cases}
\hat C(t) = \frac 1{2\gamma} (T-t)  R^\intercal \Sigma^{-1} R,\\
\hat E(t) = - \frac 1{\gamma}(T-t)  R^\intercal \Sigma^{-1} R \overline S,\\
\hat F(t) = \frac 1{4\gamma} (T-t)^2 \text{Tr}\left( R^\intercal \Sigma^{-1} R \Sigma \right) + \frac{1}{2\gamma} (T-t)  \overline S^\intercal R^\intercal \Sigma^{-1} R.
\end{cases}
$$
We consider the function $\hat \theta$ defined by
$$\hat \theta(t,S) = S^\intercal \hat C(t) S + \hat E(t)^\intercal S + \hat F(t),$$ and the associated function $\hat w$ defined by 
$$\hat w(t,\mathcal V,S) = -e^{-\gamma \left(\mathcal V + \hat \theta(t,S)\right)}.$$
For all $(t,\mathcal V,S) \in [0,T] \times \mathbb R \times \mathbb R^d $ and $q = (q_s)_{s \in [t,T]} \in \mathcal A^{Merton}_t$, we have 
\begin{align}
\label{suboptMerton}
\mathbb E \left[ -e^{-\gamma \mathcal V^{t,\mathcal V,S,q}_T} \right] \le \hat w(t,\mathcal V,S).
\end{align}
Moreover, equality is obtained in \eqref{suboptMerton} by taking the optimal control $(q^*_s)_{s \in [t,T]} \in \mathcal A^{Merton}_t$ given by the closed-loop feedback formula
\begin{align}
\label{optcontrolMerto}
q^*_s = \frac 1{\gamma } (I_d + (T-s)R^\intercal) \Sigma^{-1}R(\overline S-S_s^{t,S}).
\end{align} 
In particular, $\hat w=\hat u.$
\end{thm}

\begin{proof}
It is obvious that $(q^*_s)_{s \in [t,T]} \in \mathcal A^{Merton}_t$ (i.e., $(q^*_s)_{s \in [t,T]}$ is well defined and admissible):
\begin{align*}
	\exists C_{t,T} > 0, \forall s \in [t,T], \hspace{1cm} \parallel q^{*}_s \parallel \leqslant C_{t,T} \left( 1+\underset{\tau \in [t,s]}{\sup} \parallel S_\tau \parallel \right).
\end{align*}

Let us consider $(t,\mathcal V,S) \in [0,T] \times \mathbb R \times \mathbb R^d$ and $q = (q_s)_{s \in [t,T]} \in \mathcal A^{Merton}_t$. We now prove that 
\begin{align*}
    \mathbb E \left[\hat w\left(T,\mathcal V^{t,\mathcal V,S}_T,S^{t,S}_T\right) \right] \leqslant \hat  w(t,\mathcal V,S).
\end{align*}
We use the following notations for readability
\begin{align*}
    \forall s \in [t,T], \hspace{0.5cm} \hat w\left(s,\mathcal V^{t,\mathcal V,S,q}_s,S^{t,S}_s\right)=\hat w^{t,\mathcal V,S,q}_s,
\end{align*}
\begin{align*}
    \forall s \in [t,T], \hspace{0.5cm} \hat  \theta (s,S^{t,S}_s)=\hat \theta^{t,S}_s.
\end{align*}

By Itô's formula, we have $\forall s \in [0,T]$
\begin{align*}
    d\hat w^{t,\mathcal V,S,q}_s = \mathcal{L}^{q}\hat w^{t,\mathcal V,S,q}_s ds + \left(\partial_{\mathcal V}\hat w^{t,\mathcal V,S,q}_s q_s + \nabla_S \hat w^{t,\mathcal V,S,q}_s  \right)^{\intercal} VdW_s,
\end{align*}
where
\begin{eqnarray*}
    \mathcal{L}^{q}\hat w^{t,\mathcal V,S,q}_s&=&\partial_t w^{t,\mathcal V,S,q}_s + \left( \nabla_S   {\hat w}^{t,\mathcal V,S,q}_s\right)^\intercal R(\overline S-S) + \partial_{\mathcal V} \hat w^{t,\mathcal V,S,q}_s q^\intercal_s R(\overline S-S) + \frac 12 \text{Tr} \left( \Sigma D^2_{SS} \hat w^{t,\mathcal V,S,q}_s \right)\\
    && + \frac 12 \partial^2_{\mathcal V\mathcal V} \hat w^{t,\mathcal V,S,q}_s q_s^\intercal \Sigma q_s + \left(  \partial_{\mathcal V} \nabla_S   {\hat w}^{t,\mathcal V,S,q}_s\right)^\intercal \Sigma q_s.
\end{eqnarray*}
We have
\begin{align*}
    \nabla_S \hat w^{t,\mathcal V,S,q}_s&=-\gamma \hat w^{t,\mathcal V,S,q}_s  \nabla_S \theta ^{t,S}_s  \nonumber\\
    & = -\gamma \hat w^{t,\mathcal V,S,q}_s  \left(2 \hat C(s) S^{t,S}_s + \hat E(s) \right),
\end{align*}
and
\begin{align*}
    \partial_{\mathcal V} \hat w^{t,\mathcal V,S,q}_s&=-\gamma \hat w^{t,\mathcal V,S,q}_s.
\end{align*}
We define for all $s \in [t, T]$
\begin{align*}
    \kappa^{q}_s &= -\gamma \left(q_s + 2 \hat C(s) S^{t,S}_s + \hat E(s) \right), \\
    \xi^{q}_{t,s} &= \exp \left( \int_{t}^{s} {\kappa^{q}_\varrho}^\intercal VdW_\varrho - \frac12 \int_{t}^{s} {\kappa^{q}_\varrho}^\intercal \Sigma \kappa^{q}_\varrho d\varrho \right).
\end{align*}
We then have 
\begin{align*}
    d \left( \hat w^{t,\mathcal V,S,q}_s \left( \xi^{q}_{t,s} \right)^{-1} \right) = \left( \xi^{q}_{t,s} \right)^{-1} \mathcal{L}^{q} \hat w^{t,\mathcal V,S,q}_s ds.
\end{align*}
By definition of $\hat w$, $\mathcal{L}^{q}\hat w^{t,\mathcal V,S,q}_s \leqslant 0$. \\
\\
Moreover, equality holds for the control reaching the supremum in \eqref{HJBMerton}. It is easy to see that the supremum is reached for the unique value
\begin{align*}
    q_s &= \frac1{\gamma} \Sigma^{-1}R(\overline S-S^{t,S}_s) - \nabla_S \hat \theta(t,S^{t,S}_s)   \\
    &=\frac1{\gamma} \Sigma^{-1}R(\overline S-S^{t,S}_s) - 2 \hat C(s)S^{t,S}_s - \hat E(s)\\
    &= \frac 1{\gamma } (I_d + (T-s)R^\intercal) \Sigma^{-1}R(\overline S-S_s^{t,S})
\end{align*}
which corresponds to $(q_s)_{s \in [t,T]} = (q^*_s)_{s \in [t,T]}$.\\
\\ \\
As a consequence, $\left(\hat w^{t,\mathcal V,S,q}_s \left( \xi^{q}_{t,s} \right)^{-1}\right)_{s \in [t,T]}$ is nonincreasing and therefore
\begin{align*}
   \hat  w\left(T,\mathcal V^{t,\mathcal V,S,q}_T,S^{t,S}_T\right) \leqslant \hat w(t,\mathcal V,S) \xi^{q}_{t,T},
\end{align*}
with equality when $(q_s)_{s \in [t,T]} = (q^*_s)_{s \in [t,T]}$.
\\
\\
Taking expectation we get 
\begin{align*}
    \mathbb E \left[ \hat  w\left(T,\mathcal V^{t,\mathcal V,S,q}_T,S^{t,S}_T\right) \right] \leqslant \hat w(t,\mathcal V,S)   \mathbb E \left[\xi^{q}_{t,T}\right].
\end{align*}
We proceed to prove that $\mathbb E \left[\xi^{q}_{t,T} \right]$ is equal to 1. To do so, we use that $\xi^{q}_{t,t} = 1$ and prove that $(\xi^{q}_{t,s})_{s \in [t,T]}$ is a martingale under $\left(\mathbb P; \mathbb F = (\mathcal F_s)_{s \in [t,T]} \right)$. \\

We know that $(q^{t,q}_s)_{s \in [t,T]}$ satisfies a linear growth condition with respect to $(S^{t,S}_s)_{s \in [t,T]}$. Given the form of $\kappa$ one can easily show that there exists a constant $C$ such that
\begin{align*}
    \underset{s \in [t,T]}{\sup} {\parallel \kappa^{q}_s \parallel}^2 \leqslant C \left(1+   \underset{s \in [t,T]}{\sup} {\parallel W_s - W_t \parallel}^2  \right).
\end{align*}
By using classical properties of the Brownian motion, we prove that
\begin{align*}
    \exists \epsilon >0, \forall s \in [t,T], \hspace{0.3cm} \mathbb E \left[ \exp \left( \frac12 \int_{s}^{ \left(s+\epsilon \right) \wedge T}  {\kappa^{q}_\varrho}^\intercal \Sigma \kappa^{q}_\varrho d\varrho \right) \right] < +\infty.
\end{align*}
Using a classical trick due to Bene\v{s} (see \cite{karatzas2014brownian}, Chapter 5), we see that $(\xi^{q}_{t,s})_{s \in [t,T]}$ is a martingale under $\left(\mathbb P; \mathbb F = (\mathcal F_s)_{s \in [t,T]} \right)$. 
\\
\\
We obtain 
\begin{align*}
    \mathbb E \left[\hat  w\left(T,\mathcal V^{t,\mathcal V,S,q}_T,S^{t,S}_T\right)\right] \leqslant  \hat w(t,\mathcal V,S),
\end{align*}
with equality when $(q_s)_{s \in [t,T]} = (q^*_s)_{s \in [t,T]}$.
\\
\\
We conclude that 
\begin{align*}
    \hat u\left( t,\mathcal V,S\right) &= \underset{(q_s)_{s \in [t,T]} \in \mathcal A^{Merton}_t}{\sup} \mathbb E \left[- \exp \left(-\gamma V_T^{t,\mathcal V,S,q} \right) \right] \\
    &=\mathbb E \left[- \exp \left(-\gamma V_T^{t,\mathcal V,S,q^*} \right) \right]\\
    &=\hat w(t,\mathcal V,S).
\end{align*}

\end{proof}

\bibliographystyle{plain}

\end{document}